\documentclass[10pt,a4paper]{article}
\usepackage[utf8]{inputenc}
\usepackage[english]{babel}
\usepackage{hyperref}
\usepackage{amsmath,amsthm,enumerate}
\usepackage{amssymb}
\usepackage{tikz}
\usepackage{subfigure}
\usepackage[hmargin=2.5cm,vmargin=3cm]{geometry}
\usepackage[color=green!30]{todonotes}

\usepackage{perpage} 
\MakePerPage{footnote} 

\newtheorem{theorem}{Theorem}

\newtheorem{proposition}[theorem]{Proposition}
\newtheorem{lemma}[theorem]{Lemma}

\newtheorem{definition}[theorem]{Definition}

\newcommand{\M}{\gamma^{\text{\tiny{ID}}}}
\newcommand{\ID}{\gamma^{\text{\tiny{ID}}}}
\newcommand{\LD}{\gamma^{\text{\tiny{LD}}}}
\newcommand{\OLD}{\gamma^{\text{\tiny{OLD}}}}

\newcommand{\MD}{dim}
\newcommand{\sepID}{sep_{ID}}
\newcommand{\emp}{{\text{ID-EMP}}}
\newcommand{\univ}{{\text{ID-UNIV}}}
\newcommand{\sepLD}{sep_{LD}}
\newcommand{\empMD}{{\text{LD-EMP}}}
\newcommand{\univMD}{{\text{LD-UNIV}}}

\newcommand{\PBID}{\textsc{Identifying Code}}
\newcommand{\PBLD}{\textsc{Locating-Dominating-Set}}
\newcommand{\PBOLD}{\textsc{Open Locating-Dominating Set}}
\newcommand{\PBMD}{\textsc{Metric Dimension}}

\begin{document}

\title{Identification, location-domination and metric dimension on interval and permutation graphs. I. Bounds}
\author{Florent Foucaud\footnote{\noindent LIMOS - CNRS UMR 6158 Universit\'e Blaise Pascal, Clermont-Ferrand (France). florent.foucaud@gmail.com}
\and George B. Mertzios\footnote{\noindent School of Engineering and Computing Sciences, Durham University (United Kingdom). george.mertzios@durham.ac.uk} 
\footnote{\noindent Partially supported by the EPSRC Grant EP/K022660/1.}
\and Reza Naserasr\footnote{\noindent IRIF - CNRS UMR 8243, Université Paris Diderot, Paris, France. 
reza@irif.fr}
\and Aline Parreau\footnote{\noindent Université Lyon, CNRS, LIRIS UMR CNRS 5205, F-69621 Lyon, France}\footnote{Partially supported by a postdoctoral FNRS grant from University of Liege.}
\and Petru Valicov\footnote{\noindent LIF - CNRS UMR 7279, Université d'Aix-Marseille (France). petru.valicov@lif.univ-mrs.fr}}

\maketitle

\begin{abstract}
We consider the problems of finding optimal identifying codes, (open) 
locating-dominating sets and resolving sets of an interval or a permutation graph. 
In these problems, one asks to find a subset of vertices, normally called a
\emph{solution} set, using which all vertices of the graph are distinguished.
The identification can be done by considering the neighborhood within the solution set,
or by employing the distances to the solution vertices. Normally the goal is 
to minimize the size of the solution set then. Here we study the case of interval graphs, unit interval graphs, (bipartite) permutation graphs and cographs. For these classes of graphs we give tight lower bounds for the size of such solution sets depending on the order of the input graph. While such lower bounds for the general class of graphs
are in logarithmic order, the improved bounds in these special
classes  are of the order of either quadratic root or linear in terms of number of 
vertices. Moreover, the results for cographs lead to linear-time algorithms to solve 
the considered problems on inputs that are cographs.
\end{abstract}

\section{Introduction}

Identification problems in discrete structures are a well-studied topic. In these problems, we are 
given a graph or a hypergraph, and we wish to distinguish (i.e. uniquely identify) its vertices 
using (small) set of selected elements from the (hyper)graph.
For the \emph{metric dimension}, one seeks a set 
$S$ of vertices of a graph $G$ where every vertex of $G$ is uniquely identified by its distances to 
the vertices of $S$. The notions of \emph{identifying codes} and \emph{(open) locating-dominating 
sets} are similar; instead of the distances to $S$, we ask for the vertices to be distinguished by 
their neighbourhood within $S$. These concepts are studied by various authors since the 1970s and 
1980s, and have been applied to various areas such as network 
verification~\cite{BBDGKP11,BEEHHMR06}, fault-detection in networks~\cite{KCL98,UTS04}, graph 
isomorphism testing~\cite{B80} or the logical definability of graphs~\cite{KPSV04}. We note that the 
related problem of finding a \emph{test cover} of a hypergraph (where hyperedges are selected to 
distinguish the vertices) has been studied under several names by various authors, see 
e.g.~\cite{BS07,B72,CCCHL08,GJ79,R61}.

In this paper, we study identifying codes, (open) locatig-dominating sets and the metric dimension of interval graphs, permutation graphs and some of their subclasses. In particular, we study bounds on the order for such graphs with given size of an optimal solution.

\bigskip

\noindent\textbf{Important concepts and definitions.} All considered
graphs are finite and simple. We will denote by $N[v]$, the
\emph{closed neighbourhood} of a vertex $v$, and by $N(v)$ its
\emph{open neighbourhood} $N[v]\setminus\{v\}$. A vertex is \emph{universal} if it is adjacent to all the vertices of the graph. A set $S$ of vertices
of $G$ is a \emph{dominating set} if for every vertex $v$ of $G$, there is a
vertex $x$ in $S\cap N[v]$. It is a \emph{total dominating set} if
instead, $x\in S\cap N(v)$. In the context of (total) dominating sets
we say that a vertex $x$ \emph{(totally) separates} two distinct
vertices $u,v$ if it (totally) dominates exactly one of them. A set $S$
(totally) separates the vertices of a set $X$ if all pairs of $X$ are
(totally) separated by a vertex of $S$. Whenever it is clear from the
context, we will only say ``separate'' and omit the word
``totally''. We have the three key definitions, that merge the
concepts of (total) domination and (total) separation:

\begin{definition}[Slater~\cite{Sl87,S88}]
A set $S$ of vertices of a graph $G$ is a \emph{locating-dominating
set} if it is a dominating set and it separates the vertices of
$V(G)\setminus S$.
The smallest size of a locating-dominating set of $G$ is the \emph{location-domination number} of 
$G$, denoted $\LD(G)$. Without the domination constraint, this concept has also been used under the 
name \emph{distinguishing set} in~\cite{B80} and \emph{sieve} in~\cite{KPSV04}.
\end{definition}

\begin{definition}[Karpovsky, Chakrabarty and Levitin~\cite{KCL98}]
A set $S$ of vertices of a graph $G$ is an \emph{identifying code} if
it is a dominating set and it separates all vertices of $V(G)$.
The smallest size of an identifying code of $G$ is the
\emph{identifying code number} of $G$, denoted $\ID(G)$.
\end{definition}

\begin{definition}[Seo and Slater~\cite{SS10}]
A set $S$ of vertices of a graph $G$ is an \emph{open
locating-dominating set} if it is a total dominating set and it
totally separates all vertices of $V(G)$.
The smallest size of an open locating-dominating set of $G$ is the
\emph{open location-domination number} of $G$, denoted $\OLD(G)$. This
concept has also been called \emph{identifying open code}
in~\cite{HY12}.
\end{definition}

Separation could also be done using distances from the members of the solution 
set. Let $d(x,u)$ denote the distance between two vertices $x$ and $u$. 

\begin{definition}[Harary and Melter~\cite{HM76}, Slater~\cite{S75}] A set $B$ of vertices of a graph $G$ is a \emph{resolving set} if for each pair $u,v$ of distinct vertices, there is a vertex $x$ of $B$ with $d(x,u)\neq d(x,v)$.\footnote{Resolving sets are also known under the name of
\emph{locating sets}~\cite{S75}. Optimal resolving sets have
sometimes been called \emph{metric bases} in the literature; to
avoid an inflation in the terminology we will only use the term
\emph{resolving set}.}
The smallest size of a resolving set of $G$ is the \emph{metric dimension} of $G$, denoted 
$\MD(G)$. 
\end{definition}

It is easy to check that the inequalities $\MD(G)\leq \LD(G)\leq \ID(G)$ 
and $\LD(G)\leq \OLD(G)$ hold, indeed every
locating-dominating set of $G$ is a resolving set, and every
identifying code (or open locating-dominating set) is a
locating-dominating set. Moreover it is proved that $\ID(G)\leq
2\LD(G)$~\cite{GKM08} (using the same proof idea one would get a
similar relation between $\LD(G)$ and $\OLD(G)$ and between $\ID(G)$
and $\OLD(G)$, perhaps with a different constant factor).

In a graph $G$ of diameter~2, one can easily see that the concepts of
resolving set and locating-dominating set are almost the same, as
$\LD(G)\leq\MD(G)+1$. Indeed, let $S$ be a resolving set of $G$. Then
all vertices in $V(G)\setminus S$ have a distinct neighborhood within
$S$. There might be (at most) one vertex that is not dominated by $S$,
in which case adding it to $S$ yields a locating-dominating set.

While a resolving set and a locating-dominating set exist in every
graph $G$ (for example the whole vertex set), an identifying code may
not exist in $G$ if it contains \emph{twins}, that is, two vertices
with the same closed neighbourhood. However, if the graph is
\emph{twin-free}, then the set $V(G)$ is an identifying code of
$G$. Similarly, a graph admits an open locating-dominating set if and
only if it has no \emph{open twins}, vertices sharing the same open
neighbourhood. We say that such a graph is \emph{open twin-free}.

The focus of this work is to study these concepts and corresponding decision problems for
specific subclasses of perfect graphs. Many
standard graph classes are perfect, for example bipartite graphs,
split graphs, interval graphs. For precise definitions, we refer to
the book of Brandstädt, Le and Spinrad~\cite{BL99}. Some of these
classes are defined using a geometric intersection model, that
is, the vertices are associated to the elements of a set $S$ of
(geometric) objects, and two vertices are adjacent if and only if the
corresponding elements of $S$ intersect. The graph defined by the
intersection model of $S$ is its \emph{intersection graph}. An
\emph{interval graph} is the intersection graph of intervals of the
real line, and a \emph{unit interval graph} is an interval graph whose
intersection model contains only (open) intervals of unit
length. Given two parallel lines $B$ and $T$, a \emph{permutation graph} is the intersection graph of segments of the plane which have
an endpoint on $B$ and an endpoint on $T$. A \emph{cograph} is a graph
which can be built from single vertices using the repeated application
of two binary graph operations: the disjoint union $G\oplus H$, and
the complete join $G\bowtie H$ (another standard characterization of
cographs is that they are those graphs that do not contain a
4-vertex-path as an induced subgraph). All cographs are permutation graphs.

Interval graphs and permutation graphs are classic graph classes that
have many applications and are widely studied. They can be recognized
efficiently, and many combinatorial problems have simple and efficient
algorithms for these classes.

\bigskip

\noindent\textbf{Previous work.} It is not difficult to observe that a graph $G$ 
with $n$ vertices and an identifying code or open locating-dominating set $S$ of
size~$k$ satisfies $n\leq 2^k-1$~\cite{KCL98,SS10}. Furthermore it can be observed 
that this bound is tight. If $S$ is a locating-dominating 
set, then a tight bound is $n\leq 2^k+k-1$~\cite{S88}. These bounds are tight, even 
for bipartite graphs or split graphs. They are also tight up to a constant factor for 
co-bipartite graphs. On the other hand, tight bounds of the form $n=O(k)$ are given 
for paths and cycles~\cite{BCHL04,S88}, trees~\cite{BCHL05,Sl87} and planar graphs 
and some of their subclasses~\cite{RS84}. A bound of the form $O(k^2)$ was given for 
identifying codes in line graphs~\cite{lineID}.

The number of vertices of a graph with metric dimension~$k$ cannot be 
bounded by a function of $k$: for example, an end point of a path (of any length) 
forms a resolving set.
More generally, for every integer $k$, one can construct arbitrarily large 
trees with metric dimension~$k$ (consider for example a vertex $x$ with $k+1$
arbitrarily long disjoint paths starting from $x$). 
However, when the diameter of $G$ is at most $D$ and $\MD(G)=k$, 
we have the (trivial) bound $n\leq D^k+k$~\cite{CEJO00}, which is not tight 
but a more precise (and tight) bound is given in~\cite{HMMPSW10}.

Regarding the algorithmic study of these problems, \PBID, \PBLD, \PBOLD{} and \PBMD{} (the decision 
problems that ask, given a graph $G$ and an integer~$k$, for the existence of an identifying code, a 
locating-dominating set, an open locating-dominating set and a resolving set of size at most~$k$ in 
$G$, respectively) were shown to be NP-complete, even for many restricted graph classes. We refer to 
e.g.~\cite{A10,CHL03,DPSV12,ELW12j,F13j,lineID,GJ79,MS09,SS10} for some results. On the positive 
side, \PBID, \PBLD{} and \PBOLD{} are linear-time solvable for graphs of bounded clique-width (using 
Courcelle's theorem~\cite{CM93} ). Furthermore, Slater~\cite{Sl87} and Auger~\cite{A10} gave 
explicit linear-time algorithms solving \PBLD{} and \PBID, respectively, in trees. Epstein, Levin 
and Woeginger~\cite{ELW12j} also gave polynomial-time algorithms for the weighted version of \PBMD{} 
for paths, cycles, trees, graphs of bounded cyclomatic number, cographs and partial wheels. Diaz, 
Pottonen, Serna, Jan van Leeuwen~\cite{DPSV12} gave a polynomial-time algorithm for outerplanar 
graphs. In a companion paper~\cite{part2}, we prove that all four problems \PBID, \PBLD, \PBOLD{} 
and \PBMD{} are NP-complete, even for interval graphs and permutation graphs. We also 
give in~\cite{part2} an $f(k)poly(n)$-time (i.e. fixed-parameter-tractable) algorithm to check 
whether an interval graph has metric dimension at most~$k$.\\

\noindent\textbf{Our results and structure of the paper.} In this paper, we give new upper bounds on the maximum order of interval or permutation graphs (and some of their subclasses) having an identifying code, an (open) locating-dominating set or a resolving set of size~$k$. For the three first problems (in which the identification is neighbourhood-based), the bounds are $O(k^2)$ for interval graphs and permutation graphs and $O(k)$ for unit interval graphs, bipartite permutation graphs and cographs. We also study the metric dimension of such graphs by giving similar upper bounds in terms of the solution size~$k$ and the diameter~$D$. We obtain the bounds $O(Dk^2)$ for interval and permutation graphs, and $O(Dk)$ for unit interval graphs and cographs. We also provide constructions showing that all our bounds are nearly tight. Finally, we give a linear-time algorithm for \PBID{} and \PBOLD{} in cographs.\footnote{Remark that the algorithm of Epstein, Levin and Woeginger~\cite{ELW12j} for \PBMD{} can also be used for \PBLD.}

Section~\ref{sec:int} is devoted to interval graphs, Section~\ref{sec:unit} to unit interval graphs, Section~\ref{sec:perm} to permutation graphs, Section~\ref{sec:bip-perm} to bipartite permutation graphs, and Section~\ref{sec:cog} to cographs. We conclude the paper in Section~\ref{sec:conclu}.

\section{Interval graphs}\label{sec:int}

We now give bounds for interval graphs. Recall that in general there
are graphs with (open) location-domination or identifying code number~$k$ and
$\Theta\left(2^k\right)$ vertices~\cite{KCL98,S88}. This can be improved for
interval graphs as follows.

\begin{theorem}\label{thm:interval-sqrt}
Let $G$ be an interval graph on $n$ vertices and let $S$ be a subset
of vertices of size $k$. If $S$ is an open locating-dominating set or
an identifying code of $G$, then $n\leq \frac{k(k+1)}{2}$. If $S$ is a
locating-dominating set of $G$, then $n\leq \frac{k(k+3)}{2}$. Hence,
$\M(G)\geq\sqrt{2n+\frac{1}{4}}-\frac{1}{2}$,
$\OLD(G)\geq\sqrt{2n+\frac{1}{4}}-\frac{1}{2}$ and
$\LD(G)\geq\sqrt{2n+\frac{9}{4}}-\frac{3}{2}$.
\end{theorem}
\begin{proof}
Let $S=\{x_1,\ldots,x_k\}$ be an identifying code or open
locating-dominating set of $G$ of size $k$, where the intervals
$x_1,\ldots,x_k$ are ordered increasingly by their right endpoint (let
us denote by $r_i$, the right endpoint of interval $x_i$). Using this
order, we define a partition $\mathcal E_1,\ldots,\mathcal E_k$ of
$V(G)$ as follows. Let $\mathcal E_1$ be the set of intervals that
start strictly before $r_1$. For any $i$ with $2\leq i\leq k-1$, let
$\mathcal E_{i}$ be the set of intervals whose left endpoint lies
within $[r_{i-1},r_i[$, and let $\mathcal E_{k}$ be the set of intervals whose left endpoint is at least $r_{k-1}$. Now, let $I$ be an interval of $\mathcal E_i$ with $1\leq i\leq k$. Interval $I$ can only intersect intervals of $S$ in
$x_i,\ldots,x_k$.\footnote{We use a representation with open
intervals.} These intervals must be consecutive when considering
the order defined by the left endpoints and $I$ must intersect the
first one. There are $k-i+1$ possible intersections and so
$\mathcal E_i$ contains at most $k-i+1$ intervals. Hence, in total $G$ has at most $\sum_{i=1}^k(k-i+1)\leq\frac{k(k+1)}{2}$ vertices.

If $S$ is a locating-dominating set, we reason similarly, but we must
take into account the existence of $k$ additional vertices that do not
need to be separated (the ones from $S$).

The bounds on parameters $\M$, $\LD$ and $\OLD$ follow directly by
using the facts that $k(k+1)=(k+\frac{1}{2})^2-\frac{1}{4}$ and
$k(k+3)=(k+\frac{3}{2})^2-\frac{9}{4}$.
\end{proof}

\begin{proposition}\label{prop:interval-sqrt-tight}
The bounds of Theorem~\ref{thm:interval-sqrt} are tight for every
$k\geq 1$.
\end{proposition}
\begin{proof}For identifying codes, consider the
interval graph formed by the intersection of the following family of
intervals: $\mathcal F=\{]i,j[~|~1\leq i<j\leq k+1,i,j\in\mathbb{N}\}$, where the subfamily $\{]i,i+1[~|~1\leq i\leq k\}$ forms an identifying code $S$ of size~$k$. A similar construction can be done for open locating-dominating sets when
$k$ is even by replacing the $k/2$ intervals $]2i,2i+1[$ by
intervals $]2i-0.5,2i+0.5[$.  For locating-dominating sets,
consider $\mathcal F$ with a copy of each interval $]i,i+1[$ of
$S$. Then $S$ is a locating-dominating set.  An illustration of
these examples for $k=4$ is given in Figure~\ref{fig:extint}.
\end{proof}

\begin{figure}
\begin{center}
\tikzstyle{intervalle}=[]
\tikzstyle{intervallecode}=[line width=2pt]
\subfigure[Identifying code]{\begin{tikzpicture}[scale=1.2]
\foreach \I in {0,1,2,3}
\draw[intervallecode] (\I+0.1,0.25)--(\I+0.9,0.25);
\foreach \I in {0,1,2}
\draw[intervalle] (\I+0.1,-0.25-\I/4)--(\I+1.9,-0.25-\I/4);
\foreach \I in {0,1}
\draw[intervalle] (\I+0.1,-1-\I/4)--(\I+2.9,-1-\I/4);
\draw[intervalle] (0.1,-1.5)--(3.9,-1.5);
\end{tikzpicture}
}
\qquad
\subfigure[Open locating-dominating set]{\begin{tikzpicture}[scale=1.2]
\draw[intervallecode] (0.1,0.25)--(0.9,0.25);
\draw[intervallecode] (0.7,0)--(1.5,0);
\draw[intervallecode] (2.5,0.25)--(3.3,0.25);
\draw[intervallecode] (3.1,0)--(3.9,0);

\foreach \I in {0,1,2}
\draw[intervalle] (\I+0.1,-0.25-\I/4)--(\I+1.9,-0.25-\I/4);
\foreach \I in {0,1}
\draw[intervalle] (\I+0.1,-1-\I/4)--(\I+2.9,-1-\I/4);
\draw[intervalle] (0.1,-1.5)--(3.9,-1.5);

\end{tikzpicture}}
\qquad
\subfigure[Locating-dominating set]{\begin{tikzpicture}[scale=1.2]
\foreach \I in {0,1,2,3}
\draw[intervallecode] (\I+0.1,0.25)--(\I+0.9,0.25);
\foreach \I in {0,1,2,3}
\draw[intervalle] (\I+0.1,0)--(\I+0.9,0);
\foreach \I in {0,1,2}
\draw[intervalle] (\I+0.1,-0.25-\I/4)--(\I+1.9,-0.25-\I/4);
\foreach \I in {0,1}
\draw[intervalle] (\I+0.1,-1-\I/4)--(\I+2.9,-1-\I/4);
\draw[intervalle] (0.1,-1.5)--(3.9,-1.5);

\end{tikzpicture}}
\end{center}
\caption{\label{fig:extint}Examples of interval graphs from Proposition~\ref{prop:interval-sqrt-tight} reaching the lower bounds of Theorem~\ref{thm:interval-sqrt}. Solution intervals are in bold.}
\end{figure}
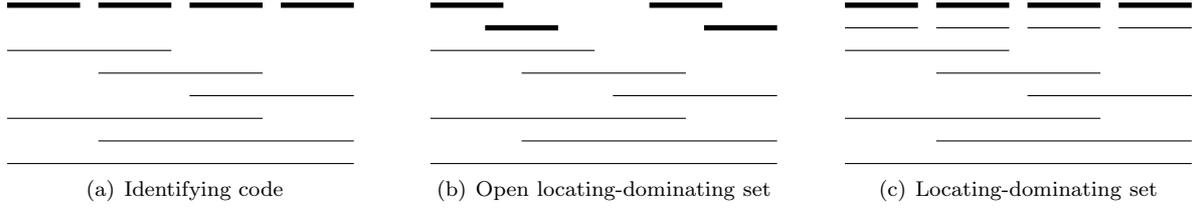

We now give a bound similar to the one of
Theorem~\ref{thm:interval-sqrt} for the metric dimension using the
diameter and the order of the graph. Recall that in general there are
graphs with metric dimension~$k$, diameter~$D$ and order
$\Theta\left(D^k\right)$~\cite{HMMPSW10}.

\begin{theorem}\label{thm:interval-MB}
Let $G$ be a connected interval graph on $n$ vertices, of diameter~$D$,
and a resolving set of size $k$. Then $n\leq 2k^2D+4k^2+kD+5k+1 = \Theta(Dk^2)$.
\end{theorem}
\begin{proof}
Let $S$ be a resolving set of size~$k$ of $G$ and let $s_1,\ldots,s_k$ be 
the elements of $S$. For each $i$ in
$\{1,\ldots,k\}$, we define an ordered set $L^i=\{x^i_1>x^i_2>\ldots>x^i_s\}$, in the
following way. Let $x^i_1$ be the left endpoint of $s_i$. Assuming
$x^i_j$ is defined, let $x^i_{j+1}$ be the smallest among all left
endpoints of the intervals of $G$ that end strictly after $x^i_j$. We stop the process 
when we have $x^i_{s+1}=x^i_s$, which means that, since $G$ is connected, 
$x^i_s$ is the smallest left endpoint among all the intervals of $G$. 
Note that an interval whose right endpoint
lies within $]x^i_{j+1},x^i_{j}]$ is at distance exactly $j+1$ of
$s_i$. Furthermore, there is no interval whose right endpoint is smaller than $x^i_{s}$.

We similarly define the ordered set 
$R^i=\{y^i_1<y^i_2<\ldots<y^i_{s'}\}$:  $y^i_1$ is the right 
endpoint of $s_i$, $y^i_{j+1}$ is the largest right endpoint among all the intervals
of $G$ that start strictly before $y^i_j$, and $y^i_{s'}$ is the largest right endpoint among all the intervals of $G$. An interval whose
left endpoint is within $[y^i_{j},y^i_{j+1}[$ is at distance exactly $j+1$
of $s_i$ and no interval has left endpoint larger than
$y^i_{s'}$.

Note that intervals at distance~1 of $s_i$ in $G$ are exactly the
intervals starting before $y^i_1$ and finishing after $x^i_1$. More
generally, for any interval of $G$, its distance to $s_i$ is uniquely
determined by the position of its right endpoint in the ordered set
$L^i$ and the position of its left endpoint in the ordered set $R^i$.
Moreover the interval $I_s$ that defines the point $x^i_s$ of $L^i$ and the interval $I_{s'}$ that defines the point $y^i_{s'}$ of $R^i$ are at distance at least $s+s'-4$ from each other. Indeed, a shortest path from $I_s$ to $I_{s'}$ contains $s_i$ or a neighbour $J$ of $s_i$. In the best case, $J$ is the interval $]x^i_2,y^i_2[$ and then $d(I_s,I_{s'})=d(I_s,J)+d(J,I_{s'})\leq s-2+s'-2$.
Therefore, we have $s+s'-4\leq D$ and $L^i\cup R^i$ contains at most $D+4$ points.

Consider now the union of all the sets $L^i\cup R^i$. Each of these sets has at most $D+4$ points 
and they all have two common points at the extremities. Thus the union contains at most $k(D+2)+2$ 
distinct points on the real line and thus defines a natural partition $\mathcal P$ of $\mathbb R$ 
into at most $k(D+2)+1$ intervals (we do not count the intervals before and after the extremities 
since no intervals can end or start there).
Any interval of $V(G)\setminus S$ is uniquely determined by the positions of its endpoints in $\mathcal P$.
Let $I\in V(G)\setminus S$. For a fixed $i$, by definition of the sets $L^i$, the interval $I$ cannot contain two points of $L^i$ and similarly, it cannot contain two points of $R^i$. Thus, $I$ contains at most $2k$ points of the union of all the sets $L^i$ and $R^i$. Therefore, if $P$ denotes a part of $\mathcal P$, there are at most $2k+1$ intervals with left endpoints in $P$.
In total, there are at most $(k(D+2)+1)\cdot (2k+1)$ intervals in $V(G)\setminus S$ and 
\begin{align*}
|V(G)| & \leq (k(D+2)+1)\cdot (2k+1)+k\\
& =2k^2D+4k^2+kD+5k+1.\hfill\qedhere
\end{align*}
\end{proof}

The bound of Theorem~\ref{thm:interval-MB} is tight up to a constant factor:

\begin{proposition}\label{prop:md-interval-tight}
For every $k\geq 1$ and $D\geq 2$, there exists an interval graph with diameter $D$, a resolving set of size~$k$, and $\Theta(Dk^2)$ vertices.
\end{proposition}
\begin{proof}
Assume that $k$ is even (a similar construction can be done if $k$ is
odd) and $D\geq 2$. Let $L>k/2$. For $i\in \{1,\ldots,k/2\}$ and $j\in \{1,\ldots,D\}$, we define the interval 
$I_{i,j}=](j-1)L+i,jL+1/2+i[$. The intervals $I_{i,j}$ for a fixed $i$ induce a path
of length $D-1$. See Figure~\ref{fig:tightmb} for an illustration with $k=6$ and $D=5$.

Let $s_i=I_{i,1}$ for $1\leq i \leq k/2$ and $s_i=I_{i-k/2,D}$ for $k/2<i\leq k$.
Using the notations of the proof of Theorem~\ref{thm:interval-MB},
one can note that, if $1\leq i \leq k/2$, then $y_j^i=jL+1/2+i$ and if $k/2<i\leq k$, 
then $x_j^i=(j-1)L+(i-k/2)$.

In particular for $1\leq i \leq k/2$ and $1<j<D$ we have: 

$$d(I_{i,j},s_{i'}) = \begin{cases}
    j-1 & \text{ if } i\leq i' \\
    j & \text{ if } i>i'\\
   \end{cases}
$$

and, for $k/2< i \leq k$ and $1<j<D$:

$$d(I_{i,j},s_{i'}) = \begin{cases}
    D-j & \text{ if } i\geq i'+k/2 \\
   D-j+1& \text{ if } i<i'+k/2.\\
   \end{cases} 
$$

Therefore, the set of intervals $\mathcal S=\{s_i , 1\leq i\leq k\}$ is a resolving set.

We add some intervals that do not influence the shortest paths between the intervals $I_{i,j}$ (in particular, the distances from $I_{i,j}$ to $\mathcal S$ do not change). 
First note that all the intervals $I_{i,j}$ have the same length. Thus there is a natural order on these intervals which is actually defined by $I_{i,j}<I_{i',j'}$ if and only if $j<j'$ or $j=j'$ and $i<i'$. In particular, any set of $k/2$ intervals that are consecutive for this order do not contain two intervals $I_{i,j}$ and $I_{i',j'}$ with $i=i'$.

Consider a particular interval $J=I_{i,j}$ with $2\leq j\leq D-2$. We add $k/2+1$ intervals after the end of $J$ in the following way. Consider the set $\{J_0<J_1<\cdots<J_{k/2}\}$ of the first $k/2+1$ intervals starting after the end of $J$. Note that $J_0$ and $J_{k/2}$ correspond to a pair of intervals $I_{i,j}$, $I_{i',j'}$ with $i=i'$. For each interval $J_s$, add an interval starting between the end of $J$ and the beginning of $J_0$ and finishing before the beginning of $J_s$ and after the beginning of $J_{s-1}$ if $s\neq 0$. See Figure~\ref{fig:tightmb} for an illustration of the intervals added in a particular example ($J=I_{3,2}$).
These intervals are all finishing before the end of $J_{k/2}$ and thus are not changing the shortest paths and the values of $x_j^i$ and $y_j^i$. 

All the intervals added this way have distinct distances to set $\mathcal S$. Indeed, either they are starting between two different consecutive pairs $y_j^i$ or finishing between different consecutive pairs $x_j^i$. There are in total $kD+(k/2+1)(D-2)k/2=\Theta(Dk^2)$ intervals in this graph and its diameter is $D$.
\end{proof}

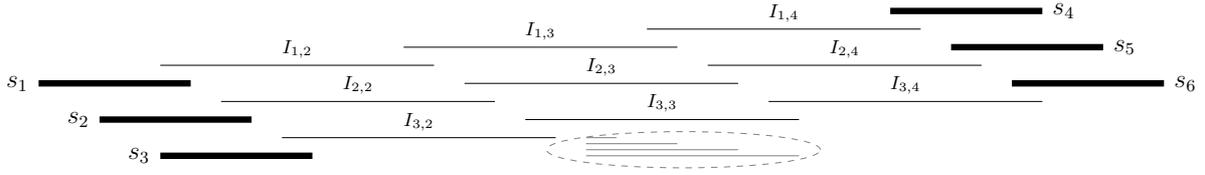
\begin{figure}[ht]
\begin{center}
\scalebox{0.8}{\begin{tikzpicture}
\tikzstyle{intervalle}=[]
\tikzstyle{intervallecode}=[line width=3pt]

\foreach \I in {1,2,3}
\draw[intervallecode] (\I+2,0.8-\I*0.6) node[left] {\large $s_{\I}$} -- (4.5+\I,0.8-\I*0.6);
\foreach \I in {4,5,6}
\draw[intervallecode] (16+\I-3,3.8-\I*0.6)--(18.5+\I-3,3.8-\I*0.6) node[right] {\large $s_{\I}$};

\foreach \I in {1,2,3}
\foreach \J in {2,3,4}
\draw[intervalle] (\J*4-4+\I,0.8-\I*0.6+0.3*\J-0.3) --node[above=0] {\small $I_{\I,\J}$} (\J*4+\I+0.5,0.8-\I*0.6+0.3*\J-0.3);

\draw[intervalle,gray] (12,-0.7)--(12.5,-0.7);
\draw[intervalle, gray] (12,-0.8)--(13.5,-0.8);
\draw[intervalle, gray] (12,-0.9)--(14.5,-0.9);
\draw[intervalle, gray] (12,-1)--(15.5,-1);

\draw[gray, dashed] (13.6,-0.9) ellipse (2.25cm and 0.3cm);
\end{tikzpicture}}
\caption{\label{fig:tightmb} An interval graph of Proposition~\ref{prop:md-interval-tight}, reaching the order of the lower bound of Theorem~\ref{thm:interval-MB}. The construction is done for diameter $D=5$ and resolving set size $k=6$. 
The intervals inside the dashed ellipse are the intervals that are added after the end of $I_{3,2}$. Similar intervals are added after the end of each interval $I_{i,j}$ for $2\leq j\leq 3$. The intervals $\{s_1,\ldots,s_6\}$ (in bold) form a resolving set.
}
\end{center}
\end{figure}

\section{Unit interval graphs}\label{sec:unit}

Using similar ideas as for Theorem~\ref{thm:interval-sqrt}
and~\ref{thm:interval-MB}, we are able to give improved bounds for unit
interval graphs.

\begin{theorem}\label{thm:lowunit}
Let $G$ be a unit interval graph on $n$ vertices and let $S$ be a
subset of vertices of size $k$. If $S$ is an open locating-dominating
set or an identifying code of $G$, then $n\leq 2k-1$. If $S$ is a
locating-dominating set of $G$, then $n\leq 3k-1$. Hence, $\M(G)\geq
\frac{n+1}{2}$, $\OLD(G) \geq \frac{n+1}{2}$, and $\LD(G)\geq
\frac{n+1}{3}$.
\end{theorem}
\begin{proof}
We consider a representation of $G$ with open unit intervals and we
denote by $\ell_I$ and $r_I$ the endpoints of the interval $I$. Consider
an identifying code or open locating-dominating set $S$ of size $k$.
Consider the set of points $T=\{\ell_I-1,\ell_I+1, \text{for all }I\in S\}$, and sort
$T$ by increasing order: $T=\{t_{1}\leq t_{2}\ldots \leq
t_{2k}\}$. Now, consider two intervals $I,I'$ such that
$\ell_I,\ell_{I'}\in~]t_{i},t_{i+1}]$. Then $I$ and $I'$ have the same
intersection under $S$. Indeed, if it is not the case, and if we
assume that $\ell_I<\ell_{I'}$, there must be an interval $I_{0}$ of
$S$ such that $\ell_I\leq r_{I_{0}}<\ell_{I'}$ ($I_{0}$ intersects $I$
but not $I'$) or such that $r_I\leq \ell_{I_{0}}<r_{I'}$ ($I_{0}$
intersects $I'$ but not $I$). But in both cases we have either
$\ell_{I_{0}}+1 \in ]t_{i},t_{i+1}[$ or $\ell_{I_{0}}-1 \in
]t_{i},t_{i+1}[$, a contradiction.

So, we must have at most one interval beginning in each period
$]t_{i},t_{i+1}]$. It is not possible to have an interval beginning
before $t_{1}$ or after $t_{i+1}$ because $S$ is also a dominating
set. Hence, there at most $2k-1$ intervals in $G$, and we are done.

By similar arguments, if $S$ is a locating-dominating set, we obtain
that there are at most $2k-1$ vertices in $V(G)\setminus S$, hence in
total at most $3k-1$ intervals.
\end{proof}

\begin{proposition}
The bounds of Theorem~\ref{thm:lowunit} are tight for every $k\geq 1$.
\end{proposition}
\begin{proof}
The bound for identifying codes is reached by odd paths
$P_{2k-1}$. Ordering its intervals $I_1,\ldots,I_{2k-1}$, the set
$S=\{I_i~|~i=1\bmod 2\}$ is an identifying code. For open
locating-dominating sets, consider a path $P_{3k-1}$ whose intervals
are ordered $I_1,\ldots,I_{3k-1}$; let $S=\{I_i~|~i=1,2\bmod 3\}$ and
add $k$ additional intervals $J_1,\ldots,J_k$, where each $J_i$ is adjacent
only to $I_{3i-2}$ and $I_{3i-1}$. It is easy to check that the resulting graph
is a unit interval graph on $4k-1$ vertices. Then $S$ is an open
locating-dominating set. For locating-dominating sets, consider the
odd path $P_{2k-1}$ and the set $S$ defined for identifying codes, and
add to this graph a copy of each interval of $S$.
\end{proof}

We also obtain the following bound for the order of a unit interval graph with a
given metric dimension and a diameter.

\begin{theorem}\label{thm:unit-MD}
Let $G$ be a connected unit interval graph on $n$ vertices, of
diameter~$D$ and with a resolving set of size~$k$. Then $n\leq k(D+2)-2$.
\end{theorem}
\begin{proof}
The proof is similar to the one of Theorem~\ref{thm:interval-MB},
except that now the right endpoint of an interval is determined
by its left endpoint. Let $s_1,\ldots,s_k$ be the elements of a resolving set $S$ of size~$k$. For
each $i$ in $\{1,\ldots,k\}$, $\ell_i$ is the left endpoint of $s_i$,
and $r_i=\ell_i+1$ is its right endpoint. Define an ordered set
$L^i=\{x^i_1>x^i_2>\ldots>x^i_s\}$, where $x^i_1=\ell_i$ (for $1<j\leq s$),
$x^i_{j+1}$ is the leftmost endpoint of an interval stopping strictly
after $x^i_{j}$ and $x^i_s=x^i_{s+1}$.  Similarly, $R^i=\{y^i_1>y^i_2>\ldots>y^i_{s'}\}$,
 with $y^i_1=r_i$, for $1<j\leq s'$, $y^i_{j+1}$ is the
rightmost endpoint of an interval starting strictly before
$y^i_{j}$ and $y^i_{s'}=y^i_{s'+1}$. In this way, the distance of an interval $I$ to $s_i$ is
determined by the position of the right endpoint of $I$ among the
points of $L^i$ and the left endpoint of $I$ among the points of $R^i$. 
Since the intervals have unit length, the position of the left endpoint of 
$I$ in $R^i$ is determined by the position of the right endpoint of $I$ in $R^i+1$
(where $R^i+1$ denotes the set $\{x+1|x \in R^i\}$). 
Therefore the distance of an interval $I$ to $s_i$ 
is determined by the position of the right endpoint of $I$ among $L_i\cup (R^i+1)$.

The distance between the leftmost and the righmost neighbor of $s_i$ is at least $s+s'-3$. Therefore, we have $|L^i\cup (R^i+1)|\leq D+3$. However, for any $i,i'$ the leftmost point of $L^i$ and
$L^{i'}$ are equal, as well as the rightmost point of $R^i$ and of
$R^{i'}$. Hence, in total, the union of all sets $L^i$ and $R^i+1$
contains at most $kD+k+2$ points, and the distance of an interval in
$V(G)\setminus S$ to elements of $S$ is determined by its position
compared to the ordering of these points. Moreover, no interval can
end before the two first points or after the two last points of $R^i+1$, so in total there are at most $kD+k-2$ possibilities. Hence $n\leq kD+k-2+k=k(D+2)-2$.
\end{proof}

Next, we show that the bound of Theorem~\ref{thm:unit-MD} is almost tight.

\begin{proposition}
For every $k\geq 1$ and $D\geq 1$, there exists a unit interval graph of diameter $D$, a resolving set of size~$k$, and $kD+1$ vertices.
\end{proposition}
\begin{proof}
For any $k,D\geq 1$ and $n=kD$, consider the $k$-th distance-power
$P^k_{kD+1}$ of a path on $kD+1$ vertices (that is, two vertices are
adjacent if and only if their distance is at most $k$ in the path
$P_{kD}$). This graph is a unit interval graph of diameter $D$. Let
$\{v_0,\ldots,v_{kD}\}$ be its vertices, ordered according the
natural order of the path. Then, the set $S=\{v_0,\ldots,v_{k-1}\}$
forms a resolving set. Indeed, for every $i,j$ with $1\leq i\leq D-1$
and $0\leq j\leq k-1$, vertex $v_{ik+j}$ is the unique vertex at
distance $i+1$ from all vertices in $\{v_0,\ldots,v_{j-1}\}$ (if
$j>0$) and at distance~$i$ from all vertices in
$\{v_{j},\ldots,v_{k-1}\}$ and vertex $v_{kD}$ is at distance $D$ from all the vertices of $S$.
\end{proof}

\section{Permutation graphs}\label{sec:perm}

We now give bounds for permutation graphs.

\begin{theorem}\label{thm:permutation-sqrt}
Let $G$ be a permutation graph on $n$ vertices and let $S$ be a $k$-subset of
$V(G)$ with $k\geq 3$. If $S$ is an open locating-dominating set
or an identifying code of $G$, then $n\leq k^2-2$. If $S$ is a
locating-dominating set of $G$, then $n\leq k^2+k-2$. Hence, $\M(G)\geq
\sqrt{n+2}$, $\OLD(G)\geq \sqrt{n+2}$ and $\LD(G)\geq
\sqrt{n+\tfrac{9}{4}}-\tfrac{1}{2}$.
\end{theorem}
\begin{proof}
Let $S$ be a set of $k$ vertices of $G$. Consider a permutation
diagram of $G$, where each vertex $v$ is represented by two integers:
the top index $t(v)$ and the bottom index $b(v)$ of its segment in the
diagram. Without loss of generality we can assume that all top indices and all bottom indices are distinct. Let
$\{t_1,\ldots,t_k\}$ and $\{b_1,\ldots,b_k\}$ be the two ordered sets
of the top and bottom indices of vertices in $S$. Now, for $1\leq
i\leq k-1$, let $\mathcal T_i$ be the set of top indices (of a vertex
of $G$) that are strictly between $t_i$ and $t_{i+1}$ in the
permutation diagram, and let $\mathcal T_{0},\mathcal T_{k}$ be, respectively, the
sets of top indices that are strictly before $t_1$ and strictly after
$t_k$. For $0\leq i\leq k$, let $\mathcal B_i$ be the
similarly defined set of bottom indices. Observe that every vertex $v$
in $V(G)\setminus S$ has its top and bottom indices $t(v)$ and $b(v)$
in some set $\mathcal T_i$ and $\mathcal B_j$, respectively.

Now, observe that the segments of two vertices $v,w$ with both
$t(v),t(w)$ in some set $\mathcal T_i$ and both $b(v),b(w)$ in some set $\mathcal B_j$
(hence both $v,w$ belong to $V(G)\setminus S$) intersect exactly the
same set of segments of $S$. Hence if $S$ is an (open)
locating-dominating set or an identifying code of $G$, we must have
$v=w$. In other words, each vertex of $V(G)\setminus S$ is uniquely
determined by the couple of intervals $(\mathcal T_i,\mathcal B_j)$ to which
its top and
bottom indices belong to. We call such a couple a {\em configuration}.

Further, let $x\in S$, $t(x)=t_i$ and $b(x)=b_j$. Then each of the two
potential vertices corresponding to the two configurations
$A_1(x)=(\mathcal T_{i-1},\mathcal B_{j-1})$ and $A_2(x)=(\mathcal
T_i,\mathcal B_j)$ are intersecting the same subset of $S$, that is the open
neighborhood of $x$, $N(x)\cap S$. Hence, if $S$ is an open
locating-dominating set, then 
any vertex has neither configuration $A_1(x)$ nor configuration $A_2(x)$
(otherwise this vertex and $x$ would not be totally separated). Also, if $S$ is a 
locating-dominating set or an identifying code, at most
one of $A_1(x)$ and $A_2(x)$ is realized. Note that, by definition, for each pair of distinct vertices $x,y \in S$, we have $A_1(x)$ and $A_1(y)$ are distinct (the same holds for $A_2(x)$ and $A_2(y)$). However we might
have $A_1(x)=A_2(y)$ for some $x\neq y$.
Nevertheless, if $S$ is an open locating-dominating set, then necessarily $A_1(x)\neq A_2(y)$,
since otherwise $x,y$ are not separated.
If $S$ is a locating-dominating set
or an identifying code, we claim that at least $k$ configurations of
the form $A_i(x)$ for $i\in\{1,2\}$ are not realized. If all of them are
distinct, we are done by the previous discussion. Otherwise, consider
a maximal sequence $x_1,\ldots,x_\ell$ of vertices of $S$ such that
$A_2(x_i)=A_1(x_{i+1})$ for every $1\leq i\leq \ell-1$. Then, these
vertices form $\ell+1$ distinct configurations of the form $A_i(x)$
for $i\in\{1,2\}$. Then, at most one such configuration can be
realized, otherwise at least two corresponding vertices would be dominated
by the same set of vertices of $S$. Repeating this argument for all
such maximal sequences yields our claim.

Similarly, the two potential vertices corresponding to configurations
$A_3(x)=(\mathcal T_{i-1},\mathcal B_{j})$ and $A_4(x)=(\mathcal
T_i,\mathcal B_{j-1})$ are intersecting in $S$ exactly the closed
neighborhood of $x$, $N[x]\cap S$. Hence, in an identifying code, no
vertex has configuration $A_3(x)$ or $A_4(x)$. In an (open)
locating-dominating set, at most one of $A_3(x)$ or $A_4(x)$ is
realized. Note that for all distinct $x,y$ in $S$, $A_3(x)\neq
A_3(y)$ and $A_4(x)\neq A_4(y)$. For identifying codes, $A_3(x)\neq
A_4(y)$, otherwise $x$ and $y$ would not be separated. For (open)
locating-dominating sets, considering again a maximal sequence of
vertices of $S$ pairwise sharing a configuration and using the same
arguments as in the previous paragraph, we get that at least $k$
configurations of the type $A_i(x)$ for $i=3,4$ are not realized.

Finally, it is clear that for any two distinct vertices $x,y$ of $S$,
$A_i(x)\neq A_j(y)$ for $i=1,2$ and $j=3,4$.

We can now proceed with counting the maximum number of realized
configurations. There are $(k+1)^2$ configurations in total.

First of all, note that the configuration $(\mathcal T_0, \mathcal
B_0)$ is not realized (otherwise the corresponding vertex would not be
dominated). For the same reason, $(\mathcal T_k, \mathcal B_k)$ is not
realized either. Moreover, $(\mathcal T_0, \mathcal B_0)$ is not a
configuration of the form $A_i(x)$ for $i\neq 1$, and $(\mathcal T_k,
\mathcal B_k)$ is not of the form $A_i(x)$ for $i\neq 2$. If
$(\mathcal T_0, \mathcal B_0)=A_1(x)$ for some $x\in S$, then $x$ is
only dominated by itself in $S$. Hence, if $S$ is an open
locating-dominating set, this cannot happen. If $S$ is an identifying
code or a locating-dominating set, consider once again the maximal
sequence $x_1,\ldots,x_\ell$ of vertices of $S$ with $x=x_1$ and
$A_2(x_i)=A_1(x_{i+1})$ for every $1\leq i\leq \ell-1$. Then, none of
the $\ell+1$ configurations of the type $A_1(x_i)$ or $A_2(x_i)$ can
be realized, since none of the corresponding vertices would be
dominated. The same argument holds if $(\mathcal T_k, \mathcal
B_k)=A_2(x')$ for some $x'\in S$. Moreover, if the two ``saved''
configurations are actually the same --- i.e. if $x=x_1$ and
$x'=x_\ell$ (then $S$ is an independent set) --- it is easy to see
that there are many additional non-realized configurations. Hence we
can assume that $(\mathcal T_0, \mathcal B_0)$ and $(\mathcal T_k,
\mathcal B_k)$ account for two additional non-realized configurations.

Now, consider the two configurations $(\mathcal T_0, \mathcal B_k)$
and $(\mathcal T_k, \mathcal B_0)$: they are both intersecting the
whole set of segments and cannot both appear, otherwise the two
corresponding vertices are not separated. If one of them is equal to
$A_i(x)$ for some $x\in S$ it must be $A_3(x)$ or $A_4(x)$. Then the
segment of $x$ is also intersecting all the segments in $S$. Hence, if
$S$ is an identifying code, none of the two configurations $(\mathcal
T_0, \mathcal B_k)$ and $(\mathcal T_k, \mathcal B_0)$ is
realized. However, in that case, we cannot have $(\mathcal T_0,
\mathcal B_k)=A_3(x)$ and $(\mathcal T_k, \mathcal B_0)=A_4(y)$,
otherwise $x,y$ would not be separated. Hence, among $(\mathcal T_0,
\mathcal B_k),(\mathcal T_k, \mathcal B_0)$ there is at least one
non-realized configuration that was not yet counted. If $S$ is an
(open) locating-dominating set, then considering a maximal sequence of
vertices of $S$ pairwise sharing a configuration and using similar
arguments as in the previous paragraph, one can show that again at
least one additional configuration is not realized.

If $S$ is an identifying code, none of the $2k$ distinct
configurations of the form $A_3(x)$, $A_4(x)$ is realized. Moreover we
also proved that at least $k$ configurations of type $A_i(x)$ for
$i\in\{1,2\}$ are not realized, and we also exhibited three additional
non-realized configurations. To summarize, we have $(k+1)^2$
configurations, from which $3k+3$ configurations are not realized,
leading to $|V(G)\setminus S|\leq k^2-k-2$ and so $|V(G)|\leq k^2-2$.

The same counting gives $|V(G)|\leq k^2-2$ if $S$ is an open
locating-dominating set, and $|V(G)|\leq k^2+k-2$ if $S$ is a
locating-dominating set.
\end{proof}

\begin{proposition}\label{prop:permutation-sqrt-tight}
The bounds of Theorem~\ref{thm:permutation-sqrt} are tight for every
$k\geq 1$.
\end{proposition}
\begin{proof}
Given $k$, one can attain the bounds using a solution $S$ inducing a path on $k$ vertices, and realizing all the configurations that were allowed in the proof of Theorem~\ref{thm:permutation-sqrt}. The key observation here is that all configurations of type $A_3$ or $A_4$ are distinct, and all configurations of type $A_1$ or $A_2$ are distinct too. See Figure~\ref{fig:tightperm} for an illustration.
\end{proof}

\begin{figure}[h]
\centering
\subfigure[Locating-dominating set]{\scalebox{0.6}{\begin{tikzpicture}[scale=0.4]
\draw[help lines] (-14.5,3) -- (15.5,3) 
      (-14.5,-3) -- (15.5,-3);
\draw[line width=3pt] (-10,-3) -- (-2,3) (-2,-3) -- (6,3) (2,-3) -- (-6,3) (10,-3) -- (2,3) ;
\draw[line width=1.1pt] (15,-3) -- (-13,3) (-6,-3) -- (-10.5,3) (6,-3) -- (10.5,3) (-0.5,-3) -- (-9,3) (0,-3) -- (0,3) (0.5,-3) -- (9,3) (-14,-3) -- (-4,3) (14,-3) -- (4,3) (-13,-3) -- (-0.6,3) (13,-3) -- (0.6,3) (-12,-3) -- (3.5,3) (12,-3) -- (-3.5,3) (4,-3) -- (-7.5,3) (-4,-3) -- (7.5,3);
\end{tikzpicture}}}\qquad
\subfigure[Identifying code]{\scalebox{0.6}{\begin{tikzpicture}[scale=0.4]
\draw[help lines] (-12.5,3) -- (12.5,3) 
      (-12.5,-3) -- (12.5,-3);
\draw[line width=3pt] (-10,-3) -- (-2,3) (-2,-3) -- (6,3) (2,-3) -- (-6,3) (10,-3) -- (2,3);
\draw[line width=1.1pt] (11,-3) -- (-12,3) (-6,-3) -- (-9.5,3) (6,-3) -- (9.5,3) (-0.5,-3) -- (-8.5,3) (0,-3) -- (0,3) (0.5,-3) -- (8.5,3) (-4,-3) -- (-1,3) (4,-3) -- (1,3) (-12,-3) -- (4,3) (12,-3) -- (-4,3);
\end{tikzpicture}}}\qquad
\subfigure[Open locating-dominating set]{\scalebox{0.6}{\begin{tikzpicture}[scale=0.4]
\draw[help lines] (-12.5,3) -- (13.5,3)
      (-12.5,-3) -- (13.5,-3);
\draw[line width=3pt] (-10,-3) -- (-2,3) (-2,-3) -- (6,3) (2,-3) -- (-6,3) (10,-3) -- (2,3) ;
\draw[line width=1.1pt] (13,-3) -- (-11,3) (-6,-3) -- (-9,3) (6,-3) -- (9,3)  (0,-3) -- (0,3)  (-12,-3) -- (-0.6,3) (12,-3) -- (0.6,3) (-11,-3) -- (3.5,3) (11,-3) -- (-3.5,3) (4,-3) -- (-7.5,3) (-4,-3) -- (7.5,3);
\end{tikzpicture}}}
\caption{\label{fig:tightperm} Examples of the constructions described in Proposition~\ref{prop:permutation-sqrt-tight} that reach the bounds in Theorem~\ref{thm:permutation-sqrt}, with $k=4$. The bold segments are part of the solution.}
\end{figure}
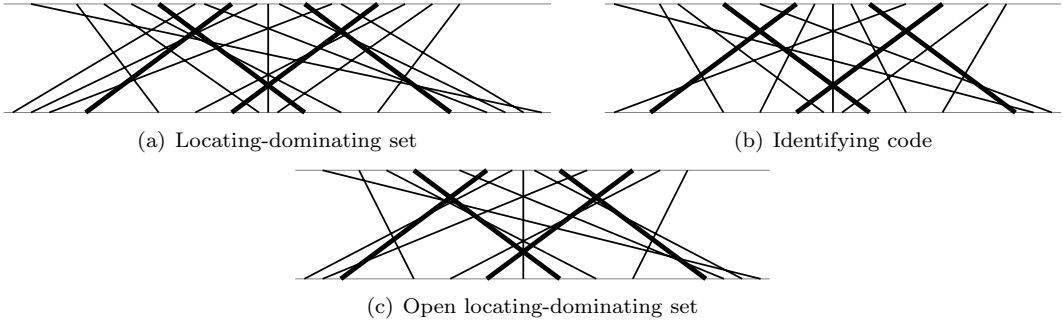

Once again, we are able to give a bound on the metric dimension of a
permutation graph in terms of its order and diameter, using ideas
similar to the ones of Theorem~\ref{thm:interval-MB}.

\begin{theorem}\label{thm:permutation-MB}
Let $G$ be a connected permutation graph on $n$ vertices, of diameter~$D$ and a resolving set of size~$k$. Then $n\leq 2k^2(D+3)+3k = \Theta(Dk^2)$.
\end{theorem}
\begin{proof}
As before, consider a permutation diagram of $G$ with top line $T$ and
bottom line $B$, where each vertex $v$ is represented by the top index
$t(v)\in T$ and the bottom index $b(v)\in B$ of its segment in the
diagram. Let $s_1,\ldots,s_k$ be the elements of a resolving set $S$ of $G$ of size~$k$. For every $i$ in
$\{1,\ldots,k\}$, we define four ordered sets $LB^i$, $LT^i$, $RB^i$,
$RT^i$ as follows. Sets $LB^i$ and $RB^i$ contain points of $B$ that
are smaller than $b(s_i)$ and greater than $b(s_i)$, respectively;
while sets $LT^i$ and $RT^i$ contain points of $T$ that are smaller
than $t(s_i)$ and greater than $t(s_i)$, respectively. More precisely,
we let $LB^i=\{lb^i_0=b(s_i),lb^i_1,\ldots,lb^i_p\}$, where for every
$j\in\{1,\ldots,p\}$, $lb^i_j=\min\{b(v)~|~d(v,s_i)=j \text{ and }
b(v)<b(s_i)\}$, that is $lb^i_j$ is the smallest bottom index of the segment of a vertex
$v$ with $d(v,s_i)=j$ and $b(v)<b(s_i)$\footnote{Note that $lb^i_1$ might not exist. In that case we just start the set $LB^i$ from $lb^i_2$.}. Similarly,
$RB^i=\{rb^i_0=b(s_i),rb^i_1,\ldots,rb^i_q\}$, where for every
$j\in\{1,\ldots,q\}$, $rb^i_j=\max\{b(v)~|~d(v,s_i)=j \text{ and }
b(v)>b(s_i)\}$. Also, $LT^i=\{lt^i_0=t(s_i),lt^i_1,\ldots,lt^i_r\}$,
where for every $j\in\{1,\ldots,r\}$, $lt^i_j=\min\{t(v)~|~d(v,s_i)=j
\text{ and } t(v)<t(s_i)\}$. Finally,
$RT^i=\{rt^i_0=t(s_i),rt^i_1,\ldots,rt^i_s\}$, where for every
$j\in\{1,\ldots,s\}$, $rt^i_j=\max\{t(v)~|~d(v,s_i)=j \text{ and }
t(v)>t(s_i)\}$.

Next, we show that the distance of any vertex $v$ of $V(G)\setminus S$
to $s_i$ is determined by the position of $t(v)$ in $LT^i$ and $RT^i$,
and the position of $b(v)$ in $LB^i$ and $RB^i$. If $t(v)>t(s_i)$ and
$b(v)<b(s_i)$ or $t(v)<t(s_i)$ and $b(v)>b(s_i)$, then $d(v,s_i)=1$.
Otherwise, assume that $v$ lies completely to the left of $s_i$:
$lt^i_{j}\leq t(v)<lt^i_{j-1}$ and $lb^i_{j'}\leq b(v)<lb^i_{j'-1}$
for some $j,j'$ with $1\leq j\leq r$ and $1\leq j'\leq p$ (the case
where it lies to the right of $s_i$ is symmetric). Then, we claim that
$d(v,s_i)=\min\{j+1,j'+1\}$. If $j\leq j'$, by definition of $j$, the
segment of $v$ cannot intersect any segment with distance at
most~$j-1$ to $s_i$, hence $d(v,s_i)\geq j+1$. However, the segment
whose top endpoint is $lt^i_{j}$ must intersect a segment with
distance~$j-1$ to $s_i$, hence it also crosses the segment of $v$, and
$d(v,s_i)\leq j+1$. If $j'<j$, a similar argument holds.

Now, since $G$ has diameter $D$, we have $|LB^i\cup RB^i|=p+q+1\leq
D+3$, and $|LT^i\cup RT^i|=r+s+1\leq D+3$. Indeed, consider the
shortest path $P_l$ of length~$p$ starting from the vertex whose
bottom index is $lb_p$ and goes to $s_i$. Consider a similar shortest
path $P_r$ of length~$q$ from $s_i$ to the vertex whose bottom index
is $rb_q$. If the concatenation of these paths is a shortest path, we
are done since in this case $p+q\leq D$. Otherwise, notice that a
shortcut can only exist around $s_i$. In fact, it could only be that
the penultimate vertex of $P_l$ and the second vertex of $P_r$ are
adjacent, or that the ante-penultimate of $P_l$ and the third vertex of
$P_r$ are both adjacent to a neighbor of $s_i$. In any case, the
resulting shortest path has length at least $p+q-2$ and at most $D$,
hence $p+q+1\leq D+3$, as claimed.

It follows that using the union of all sets $LT^i$ and $RT^i$ (respectively, $LB^i$ and $RB^i$), $1\leq i\leq k$, induces a partition $\mathcal P(T)$ of the line $T$ (respectively, $\mathcal P(B)$ of the line $B$) into at most $k(D+3)+1$ parts. Moreover, for any vertex $v$ in $V(G)\setminus S$, the membership of $b(v)$ in a given part of $\mathcal P(B)$ and of $t(v)$ in a given part of $\mathcal P(T)$ completely determines the distances of $v$ to the vertices in $S$. Let $v$ be a vertex of $V(G)\setminus S$, with $b(v)$ belonging to some part $P$ of the partition $\mathcal P(B)$. For a given $i$ ($1\leq i\leq k$), by definition of $LT^i$ and $RT^i$, there are only two possibilities for the points of $LT^i\cup RT^i$ that $t(v)$ lies between. Hence, there are at most $2k$ vertices in $V(G)\setminus S$ whose associated segment has its bottom point within part $P$ of $\mathcal P(B)$. Hence, in total we have
\begin{align*}
|V(G)| & \leq (k(D+3)+1)\cdot 2k+k\\
& = 2k^2(D+3)+3k,
\end{align*}
which completes the proof.
\end{proof}

We do not know if the bound of Theorem~\ref{thm:permutation-MB} is tight, but we are able to provide a construction similar to the one for interval graphs of Proposition~\ref{prop:md-interval-tight}, showing that this bound is almost tight.

\begin{proposition}\label{prop:md-perm-tight}
For every even $k\geq 2$ and every $D\geq 2$, there are permutation graphs of diameter $D$, a resolving set of size~$k$ and on $\Theta(Dk^2)$ vertices.
\end{proposition}
\begin{proof}
Consider $k/2$ paths $P_1, \ldots, P_{k/2}$ of length $D-1$ where the path $P_{i+1}$ is obtained from path $P_i$ by a translation (see Figure~\ref{fig:mdtightperm}). Let $P_i=\{u^i_1,\ldots,u^i_D\}$. The endpoints of the paths (i.e. the vertices $u^i_1$ and $u^i_D$) form a resolving set. One can add $k/2+2$ vertices that have the bottom index lying between the bottom points of two consecutive segments $u^i_{2k}$ and $u^i_{2k+1}$ of the same path (see the figure). In this way, we add $k/2+2$ segments for each of the $D/2$ intersections of the $k/2$ paths. In the end, the graph has $\Theta(Dk^2)$ vertices.
\end{proof}

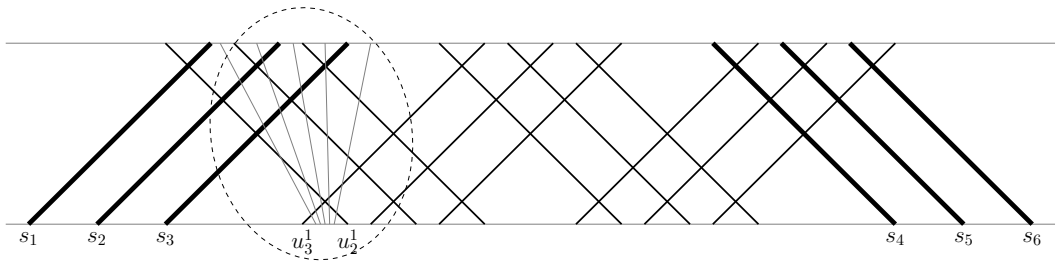
\begin{figure}[!ht]
\centering
\scalebox{0.6}{\begin{tikzpicture}

\draw[help lines] (-10.5,2) -- (12.5,2) 
      (-10.5,-2) -- (12.5,-2);

\foreach \I in {1,2,3}{
\draw[line width=3pt] (-10+\I*1.5-1.5,-2) node[below] {\Large $s_{\I}$} -- (-6+\I*1.5-1.5,2) ;}
\foreach \I in {0,1,2}{
\draw[line width=1.1pt] (-7+\I*1.5,2) -- (-3+\I*1.5,-2) node(\I1) {};
\draw[line width=1.1pt] (-4+\I*1.5,-2) node(\I2) {} -- (0+\I*1.5,2) ;
\draw[line width=1.1pt] (-1+\I*1.5,2) -- (3+\I*1.5,-2);
\draw[line width=1.1pt] (2+\I*1.5,-2) -- (6+\I*1.5,2) ;}
\foreach \I in {4,5,6}{
\draw[line width=3pt] (5+\I*1.5-6,2) -- (9+\I*1.5-6,-2)node[below] {\Large $s_{\I}$};
}

\draw[gray] (-3.7,-2) -- (-5.8,2);
\draw[gray] (-3.6,-2) -- (-5,2);
\draw[gray] (-3.5,-2) -- (-4.2,2);
\draw[gray] (-3.4,-2) -- (-3.5,2);
\draw[gray] (-3.3,-2) -- (-2.5,2);

\draw (01) node[below] {\Large $u^1_{2}$};
\draw (02) node[below] {\Large $u^1_{3}$};

\draw[dashed,rotate around={10:(-3.8,0)}] (-3.8,0) ellipse (2.2cm and 2.8cm);
\end{tikzpicture}}
\caption{\label{fig:mdtightperm} An example of the construction of Proposition~\ref{prop:md-perm-tight} with diameter $D=6$, a resolving set of size $k=6$, and order $\Theta(Dk^2)$. 
Segments similar to the ones in the dashed ellipse are added between the bottom points of every two consecutive segments $u^i_{2j}$ and $u^i_{2j+1}$ of path $P_i$, for $1\leq i\leq k/2$ and $1\leq j\leq D/2$. The bold segments $\{s_1,\ldots,s_6\}$ form a resolving set.}
\end{figure}

\section{Bipartite permutation graphs}\label{sec:bip-perm}

A graph is a bipartite permutation graph if it is a permutation graph
and it is bipartite. A characterization due to
Spinrad~\cite{Sp87} uses orderings of the vertices, as follows. Let $G$ be a
bipartite graph with parts $A$ and $B$. An ordering $<$ of $A$ has the
\emph{adjacency property} if, for every vertex $b\in B$, its
neighbourhood $N(b)$ consists of vertices that are consecutive in
$<$. It has the \emph{enclosure property} if, for every pair $b,b'$ of
vertices in $B$ with $N(b) \subseteq N(b')$, the vertices of
$N(b')\setminus N(b)$ are consecutive in $<$. A bipartite graph $G$
with parts $A$ and $B$ is a bipartite permutation graph if and only if
it admits an ordering of $A$ that has the adjacency and enclosure
properties.

\begin{theorem}\label{thm:bip-perm}
Let $G$ be a bipartite permutation graph on $n$ vertices and let $S$
be a $k$-subset of $V(G)$. If $S$ is a locating-dominating
set or an identifying code, then $n\leq 3k+2$. If $S$ is an open
locating-dominating set, then $n\leq 2k+2$. Hence, $\M(G)\geq
\frac{n-2}{3}$, $\OLD(G) \geq \frac{n-2}{2}$, and $\LD(G)\geq
\frac{n-2}{3}$.
\end{theorem} 
\begin{proof}
Let $G$ be a bipartite permutation graph with parts $A$ and $B$. We
may assume that $|A|\geq 2$ and $|B|\geq 2$ (otherwise the graph is a
star with isolated vertices and the bounds hold, indeed
$\M(K_{1,k})=\LD(K_{1,k})=k$; for $k\geq 2$, $K_{1,k}$ has no open
locating-dominating set, and $\OLD(K_{1,1})=2$).

Let $<_A$ be an ordering of $A$ that has the adjacency and enclosure
properties. We also order the vertices of $B$ using the natural
ordering $<_B$ of their neighbourhoods within $A$ along $<_A$: given
$b_1,b_2\in B$ with $x_1,y_1$ and $x_2,y_2$ the smallest and largest
(according to $<_A$) members of $N(b_1)$ and $N(b_2)$, respectively,
we have $b_1<_B b_2$ if $x_1<_A x_2$ or $x_1=x_2$ and $y_1\leq y_2$. 
Note that $<_B$ has the adjacency and the enclosure properties with respect to $<_A$.
Indeed, let $a\in A$ and consider its smallest and largest neighbors $b_1$ and $b_2$ (then $b_1\leq_B b_2$). Let $b$ be an element between $b_1$ and $b_2$. Let $x_b$ (respectively $y_b$) be the smallest (resp. largest) neighbor in $A$ of $b$. Since $a$ is adjacent to $b_2$ and $b<_Bb_2$, we have $x_b\leq_A a$. Similarly $a\leq_A y_b$. Since $<_A$ has the adjacency property, $a$ is in the neighborhood of $b$ and the neighborhood of $a$ consists of vertices that are consecutive in $<_B$ and $<_B$ has the adjacency property. Consider now two vertices $a$ and $a'$ of $A$ with $N(a)\subseteq N(a')$. Without loss of generality, we can assume that $a'<a$. Assume there exists a vertex $b$ in $N(a')\setminus N(a)$ that is larger than all the vertices of $N(a)$. Let $b'$ be an element of $N(a)$. We have $b'<_B b$. Let $x_{b}$ and $x_{b'}$ be the smallest elements of $N(b)$ and $N(b')$ respectively, and $y_b$ and $y_{b'}$ be the largest elements of $N(b)$ and $N(b')$ respectively. 
Then $y_b<_A a \leq_A y_{b'}$. Since $b'<_B b$, we must have $x_{b'}<_A x_b$ and we get a contradiction because $<_A$ has the enclosure property, $N(b)\subseteq N(b')$ but $x_2$ and $a$ are in $N(b')\setminus N(b)$ with some elements of $N(b)$ in between. Hence $<_B$ has the enclosure property. Note the order induced on $A$ by $<_B$ (as we defined $<_B$ from $<_A$) is exactly $<_A$.

Consider the set $\mathcal P_B$ of all pairs of consecutive vertices
in $B$ with respect to $<_B$. We have $|\mathcal P_B|=|B|-1$. Let $S$ be
an (open) locating-dominating set or an
identifying code. We claim that a vertex $a$ of $A\cap S$ can separate
at most two pairs in $\mathcal P_B$. Indeed, a pair $b,b'$ in $\mathcal P_B$ is
separated by $a$ if and only if $a$ is adjacent to exactly one of
$b,b'$. Let the vertices of $B$ be ordered $b_1 <_B\ldots<_B
b_{|B|}$. Then, by definition of $<_A$ and $<_B$, there are indices
$1\leq \ell,r\leq |B|$ with $a,b_i$ non-adjacent if $i\leq\ell$ or
$i\geq r$, and $a,b_i$ adjacent if $\ell<i<r$. The claim follows.

Now, let $b=b_i$ belong to $B\cap S$. Then $b$ may only separate
two pairs in $\mathcal P_B$: the ones $b$ itself belongs to, i.e., 
$\{b_{i-1},b_i\}$ and $\{b_{i},b_{i+1}\}$. However, we claim that the vertices of
$S_B=S\cap B$ account only for at most $|S_B|$ pairs in $S_B$ that are
\emph{only} separated by vertices of $S_B$. Indeed, let $b_\ell <_B
\ldots <_B b_r$ be a maximal sequence of vertices of $B$ that belong
to $S_B$. Then these $r-\ell+1$ vertices separate altogether at most
$r-\ell+2$ pairs of $\mathcal P_B$. If they separate exactly $r-\ell+2$ such
pairs, then $\ell>1$ and $r<|B|$. But the pair
$\{b_{\ell-1},b_{r+1}\}$ is also separated by $S$, that is, by a
vertex $a$ in $A\cap S$. But then one of the pairs of $\mathcal P_B$ separated
by $b_\ell, \ldots,b_r$ is also separated by $a$. Hence there are at
most $r-\ell+1$ pairs of $\mathcal P_B$ separated by these $r-\ell+1$
vertices. Repeating this counting for all such maximal subsequences
yields the claim.

Moreover, observe that if $S$ is an open locating-dominating set, a
vertex $b$ in $S\cap B$ cannot separate any pair in $\mathcal P_B$.

We can use the same arguments reversing the role of $<_B$ and $<_A$ and the set $\mathcal P_A$ of pairs of consecutive vertices of $A$. 

To summarize, if $S$ is a locating-dominating set or an identifying
code, the vertices in $S\cap A$ may separate at most $2|S\cap A|$
pairs of $\mathcal P_B$ and separate at most $|S\cap A|$ pairs of $\mathcal P_A$ that
are not separated by a vertex of $S\cap B$, and vice-versa the
vertices in $S\cap B$ may separate at most $2|S\cap B|$ pairs of $\mathcal P_A$
and separate at most $|S\cap B|$ pairs of $\mathcal P_B$ that are not separated
by a vertex of $S\cap A$. Hence $|B|-1=|\mathcal P_B|\leq 2|S\cap A|+|S\cap B|$
and $|A|-1=|\mathcal P_A|\leq 2|S\cap B|+|S\cap A|$. In total, $n=|A|+|B|\leq
3|S|+2$ and we are done.

Similarly, if $S$ is an open locating-dominating set, we have
$|B|-1=|\mathcal P_B|\leq 2|S\cap A|$ and $|A|-1=|\mathcal P_A|\leq 2|S\cap B|$,
yielding $n\leq 2|S|+2$.
\end{proof}

The bounds of Theorem~\ref{thm:bip-perm} are almost tight. 

\begin{proposition}\label{prop:low-bip-perm}
For every $k\geq 1$, there exist three bipartite permutation graphs with
a locating-dominating set, an identifying code and an open
locating-dominating set of size $3k-1$, $3k-3$ and $2k-2$, respectively.
\end{proposition}
\begin{proof}
For location-domination, consider an odd path $P_{2k-1}$ with
$V(P_{2k-1})=\{v_0,\ldots,v_{2k-2}\}$, let $S=\{v_i~|~i=0\bmod 2\}$
and attach a pendant vertex to every vertex in $S$. This graph is a
bipartite permutation graph (observe that $S$, together with its natural ordering,
has the adjacency and enclosure properties), it has $n=3k-1$ vertices and $S$ is a
locating-dominating set.

For identifying codes, again select the odd path $P_{2k-1}$ with
$S=\{v_i~|~i=0\bmod 2\}$, but now for every $i\in\{2,\ldots, 2k-4\}$
add a vertex adjacent to $\{v_i,v_{i+2},v_{i+4}\}$. Again $S$ is an
identifying code and the graph has $n=3k-3$ vertices.

For open locating-dominating sets, select any path $P_k$ with
$S=V(P_k)=\{v_0,\ldots,v_{k-1}\}$, and attach a pendant vertex to
every vertex in $S\setminus\{v_1,v_{k-2}\}$. Again $S$ is an
open locating-dominating set and the graph has $n=2k-2$ vertices.

The constructions are illustrated in Figure~\ref{fig:constr-bip-perm}.
\end{proof}

\begin{figure}[!htpb]
  \centering
  \subfigure[Locating-dominating set]{
    \scalebox{1.0}{
      \begin{tikzpicture}[join=bevel,inner sep=0.5mm,scale=0.75,line width=0.5pt]

        \foreach \i [evaluate=\i as \e using \i*2, evaluate=\i as \o using \i*2+1] in {0,...,5}{

          \pgfmathtruncatemacro\e{\e}
          \pgfmathtruncatemacro\o{\o}

          \path (\i,0) node[draw,shape=circle,fill=black] (\e) {};
          \draw (\e) node[above left=0.05cm] {$v_{\e}$};
          \path (\i,1) node[draw,shape=circle] (\e-t) {};

          \draw (\e)--(\e-t);

          \ifnum \i<5
          \path (\i+0.5,-1) node[draw,shape=circle] (\o) {};
          \draw (\o) node[below=0.1cm] {$v_{\o}$};
          \draw (\e)--(\o);
          \fi

          \ifnum \i>0
          \pgfmathtruncatemacro\o{\o-2}
          \draw (\e)--(\o);
          \fi
        }

      \end{tikzpicture}
    }
  }\qquad
  \subfigure[Identifying code]{
    \scalebox{1.0}{
      \begin{tikzpicture}[join=bevel,inner sep=0.5mm,scale=0.75,line width=0.5pt]

        \foreach \i [evaluate=\i as \e using \i*2, evaluate=\i as \o using \i*2+1] in {0,...,5}{

          \pgfmathtruncatemacro\e{\e}
          \pgfmathtruncatemacro\o{\o}

          \path (\i,0) node[draw,shape=circle,fill=black] (\e) {};
          \draw (\e) node[left=0.1cm] {$v_{\e}$};
          
          \ifnum \i>0
          \ifnum \i<5
          \path (\i,1) node[draw,shape=circle] (\e-t) {};
          \pgfmathtruncatemacro\a{\e-2}
          \pgfmathtruncatemacro\b{\e+2}
          \draw (\e)--(\e-t)--(\a) (\b)--(\e-t);
          \fi
          \fi
                    
          \ifnum \i<5
          \path (\i+0.5,-1) node[draw,shape=circle] (\o) {};
          \draw (\o) node[below=0.1cm] {$v_{\o}$};
          \draw (\e)--(\o);
          \fi
          
          \ifnum \i>0
          \pgfmathtruncatemacro\o{\o-2}
          \draw (\e)--(\o);
          \fi

        }
      \end{tikzpicture}
    }
  }\qquad
  \subfigure[Open locating-dominating set]{
    \scalebox{1.0}{
      \begin{tikzpicture}[join=bevel,inner sep=0.5mm,scale=0.75,line width=0.5pt]

        \foreach \i [evaluate=\i as \e using \i*2, evaluate=\i as \o using \i*2+1] in {0,...,2}{

          \pgfmathtruncatemacro\e{\e}
          \pgfmathtruncatemacro\o{\o}

          \path (\i*1.5,0) node[draw,shape=circle,fill=black] (\e) {};
          \draw (\e) node[above=0.1cm] {$v_{\e}$};
         
          \ifnum \i<2
          \path (\i*1.5,-1) node[draw,shape=circle] (\e-b) {};
          \draw (\e)--(\e-b);
          \fi         
          
          \path (\i*1.5+0.75,-1) node[draw,shape=circle,fill=black] (\o) {};
          \draw (\o) node[below=0.1cm] {$v_{\o}$};
          \draw (\e)--(\o);
 
          \ifnum \i>0
          \path (\i*1.5+0.75,0) node[draw,shape=circle] (\o-t) {};
          \draw (\o)--(\o-t);
          \fi

          \ifnum \i>0
          \pgfmathtruncatemacro\o{\o-2}
          \draw (\e)--(\o);
          \fi

          \path (-0.5,0) node {};
          \path (4.5,0) node {};

        }
      \end{tikzpicture}
    }
  }

\caption{Extremal constructions of Proposition~\ref{prop:low-bip-perm}
with $k=6$. Black vertices belong to the
solution.}\label{fig:constr-bip-perm}
\end{figure}
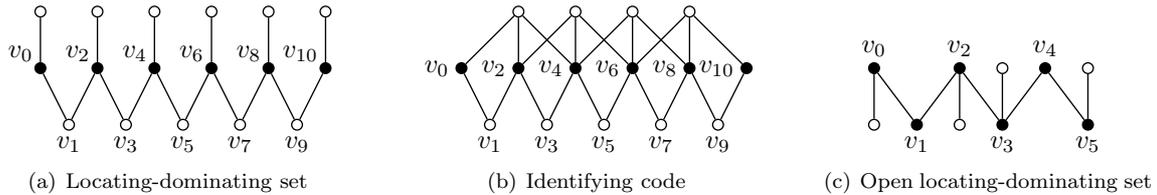

\begin{theorem}\label{thm:bip-permutation-MD}
Let $G$ be a connected bipartite permutation graph on $n$ vertices, of diameter~$D$ and a resolving set of size~$k$. Then $n\leq k(2D-1)+2 = \Theta(Dk)$.
\end{theorem}
\begin{proof}
Let $A$ and $B$ be the two parts of the bipartition of $G$, and consider two orderings $<_A$ of $A$ and $<_B$ of $B$ that have the adjacency and enclosure properties. Let $S=\{s_1,\ldots,s_k\}$ be a resolving set of $G$, and assume without loss of generality that for some $i\in\{1,\ldots,k\}$, $s_i\in A$. Then, the sets $A$ and $B$ can be partitioned into parts consisting of consecutive vertices (with respect to $<_A$ and $<_B$), where the vertices in each part have the same distance to $s_i$. Moreover, the vertices in $A$ have even distances to $s_i$, while the vertices in $B$ have odd distances to $s_i$. The number of parts in $A$ and in $B$ is at most $D$.

 Repeating this process for each vertex of $S$, we have partitioned the vertices in $A\setminus S$ and of $B\setminus S$ into at most $k(D-1)+1$ parts each. Each part may contain at most one vertex of $V(G)\setminus S$ (since the membership in a part determines the distances to the vertices of $S$). Hence, we have
\begin{align*}
|V(G)| & \leq 2(k(D-1)+1)+k\\
& = k(2D-1)+2.\hfill\qedhere
\end{align*}
\end{proof}

Next, we show that Theorem~\ref{thm:bip-permutation-MD} is asymptotically tight.

\begin{proposition}\label{prop:md-bip-perm-tight}
For every even $k\geq 2$ and every $D\geq 2$, there exists a bipartite permutation graph of diameter $D$, a resolving set of size~$k$ and $\Theta(Dk)$ vertices.
\end{proposition}
\begin{proof}
Simply consider the construction of Proposition~\ref{prop:md-perm-tight} built from $k/2$ paths of length $D-1$ each (omitting the second part of the construction of Proposition~\ref{prop:md-perm-tight}). This bipartite permutation graph has $Dk/2$ vertices and a resolving set of size~$k$.
\end{proof}

\section{Algorithm and bounds for Cographs}\label{sec:cog}

The \emph{cotree} of a cograph $G$ is a tree where the leaves are the
vertices of $G$, and the inner nodes are of type $\oplus$ and
$\bowtie$. This tree represents the construction of $G$ using the two
operations. A cograph can be recognized in linear time and its
corresponding cotree can be constructed in linear time
too~\cite{HP05}. Many problems can be computed in linear time in
cographs using their cotree representation and by simple bottom-up
computation. Epstein, Levin and Woeginger~\cite{ELW12j} gave such an
algorithm for computing the metric dimension. Observe that connected
cographs have diameter at most~2, hence, as already discussed, for any
connected cograph $G$ we have $\MD(G)\leq\LD(G)\leq\MD(G)+1$, and
$\LD(G)=\MD(G)+1$ if and only if every minimum resolving set has a
non-dominated vertex. The latter fact is computed by the algorithm
of~\cite{ELW12j}, which can therefore be used for computing the
location-domination number of a cograph.

In this subsection, we will give a similar linear-time algorithm for
\PBID, and we will give linear lower bounds on the value of parameters
$\LD$, $\MD$, $\ID$ and $\OLD$ in terms of the order.

\subsection{The algorithm}

We describe in detail the algorithm for identifying codes (the one for
open locating-dominating sets is very similar). We denote by
$\sepID(G)$ the smallest size of a \emph{separating set}, that is, a set
$S\subseteq V(G)$ where for every pair $u,v$ of distinct vertices,
there is an element of $S$ dominating exactly one of $u,v$ (it is an
identifying code without the condition of being a dominating set). It
follows from the definitions that
$\sepID(G)\leq\M(G)\leq\sepID(G)+1$, where the upper bound is reached if
and only if for every smallest separating set there is a non-dominated
vertex.
Therefore, if we can compute $\sepID(G)$ as well as decide the latter
fact, then we can compute $\M(G)$.

We define $\emp(G)$ to be the property that for a graph $G$, every
minimum separating set leaves a non-dominated vertex of $G$; $\univ(G)$ is the
property that for every minimum separating set $S$ of $G$, there exists a vertex of $G$
that is dominated by all vertices of $S$.

The advantage of using $\sepID(G)$ comes from the following lemma. 

\begin{lemma}\label{prop:induct_sep}
Let $G_1,G_2$ be two twin-free graphs with $\sepID(G_1)=k_1$ and
$\sepID(G_2)=k_2$. Then, $k_1+k_2\leq\sepID(G_1\oplus G_2)\leq k_1+k_2+1$,
where the upper bound is reached if and only if properties $\emp(G_1)$
and $\emp(G_2)$ hold. Moreover, suppose $G_1\bowtie G_2$ is a twin-free graph, then $k_1+k_2\leq\sepID(G_1\bowtie G_2)\leq
k_1+k_2+1$ and the upper bound is reached if and only if properties
$\univ(G_1)$ and $\univ(G_2)$ hold.
\end{lemma}
\begin{proof}
Note that in both $G_1\oplus G_2$ and $G_1\bowtie G_2$, a vertex in
$G_1$ cannot separate a pair in $G_2$, and vice-versa. Hence, for
every separating set of $G_1\oplus G_2$ or $G_1\bowtie G_2$, its
restriction to $G_i$ for $i\in\{1,2\}$ is a separating set of
$G_i$. This proves the two lower bounds.

For the upper bounds, let $S_1$ and $S_2$ be minimum separating sets
of $G_1$ and $G_2$, respectively. If $S=S_1\cup S_2$ is not a
separating set of $G_1\oplus G_2$ or $G_1\bowtie G_2$, by the previous
observation, there must be a pair $u,v$ with $u\in G_1$ and $v\in G_2$
that is not separated. In the case of $G_1\oplus G_2$, these two
vertices must both be non-dominated by $S$, and this is the only
non-separated pair. Then, adding one of them gives a separating set of
size $k_1+k_2+1$. For the case $G_1\bowtie G_2$, $u$ is dominated
by all vertices of $S_2$, and $v$ is dominated by all vertices of
$S_1$. Hence both $u,v$ must be dominated by all vertices of $S$, and
this is the only non-separated pair. Since $u,v$ are not twins, there
must be a vertex $w$ that separates them; $S\cup\{w\}$ is a separating
set of $G_1\bowtie G_2$ of size $k_1+k_2+1$.
\end{proof}

Using the following lemma, it is easy to keep track of the properties
$\emp$ and $\univ$ while parsing the cotree structure of a cograph
$G$.

\begin{lemma}\label{prop:emp-univ} We have:\\
1. $\emp(K_1)$ and $\univ(K_1)$;\\
2. $\emp(G_1\oplus G_2)$ if and only if $\emp(G_1)$ or $\emp(G_2)$;\\
3. $\univ(G_1\oplus G_2)$ if and only if one of $G_1,G_2$ (say $G_1$) is $K_1$, $\univ(G_2)$ and $\neg\emp(G_2)$;\\
4. $\emp(G_1\bowtie G_2)$ if and only if $G_1=K_1$, $\neg\univ(G_2)$ and $\emp(G_2)$;\\
5. $\univ(G_1\bowtie G_2)$ if and only if $\univ(G_1)$ or $\univ(G_2)$.\\
6. If $\neg \emp(G)$ and $\neg \univ(G)$, then there exists a minimum separating set $S$ of $G$ such that every vertex of $G$ is dominated by some element of $S$ but no vertex of $G$ is dominated by the entire set $S$.
\end{lemma}

\begin{proof}
We prove the lemma by induction using the cotree structure of cographs.

The first item is clearly true.
For the second item, assume $\emp(G_1\oplus
G_2)$. By Lemma~\ref{prop:induct_sep}, if $\neg\emp(G_1)$ or
$\neg\emp(G_2)$, then any minimum separating set of $G_1\oplus G_2$ is
the union of a minimum separating set of $G_1$ and one of $G_2$. Hence
if both $\neg\emp(G_1)$ and $\neg\emp(G_2)$, then $\neg\emp(G_1\oplus
G_2)$, which is a contradiction. Now, if $\emp(G_1)$ and $\neg\emp(G_2)$ (or vice-versa), by Lemma~\ref{prop:induct_sep}, we have $\emp(G_1\oplus G_2)$. If both $\emp(G_1)$ and
$\emp(G_2)$, then again by Lemma~\ref{prop:induct_sep},
$\sepID(G_1\oplus G_2)=\sepID(G_1)+\sepID(G_2)+1$, but since no vertex of
$G_1$ dominates any vertex of $G_2$ (and vice-versa), there must
remain a non-dominated vertex in $G_1\oplus G_2$.

For the third item, assume $\univ(G_1\oplus G_2)$. If none of
$G_1,G_2$ is $K_1$, then there must be a code vertex in both
$G_1,G_2$, which would imply that $\neg\univ(G_1\oplus G_2)$ and contradict the hypothesis. Thus we may assume $G_1=K_1$,
and let $S_2$ be a minimum separating set of $G_2$. By
Lemma~\ref{prop:induct_sep}, if $\emp(G_2)$, $\sepID(G_1\oplus
G_2)=\sepID(G_1)+\sepID(G_2)+1$. But then $S'=S_2\cup V(K_1)$ is a minimum
separating set of $G_1\oplus G_2$ without a vertex dominated by the
whole of $S'$, a contradiction. Hence, $\neg\emp(G_2)$. If
we also have $\neg\univ(G_2)$, by induction hypothesis and using item $6$, there exists a minimum separating set $S_2$ of $G_2$ with no vertex dominated by the whole set $S_2$ and with all vertices of $G_2$ dominated by $S_2$. 
Hence $S_2$ is a minimum separating set of $G_1\oplus G_2$ without a vertex dominated by the whole set $S_2$ and we have $\neg \univ(G_1\oplus G_2)$, a contradiction. 
For the converse, if $G_1=K_1$, $\univ(G_2)$ and
$\neg\emp(G_2)$, then by Lemma~\ref{prop:induct_sep}, $\sepID(G_1\oplus
G_2)=\sepID(G_1)+\sepID(G_2)$, and it is clear that there is no minimum separating set of $G_1\oplus G_2$ containing the vertex of $K_1$.
Hence every minimum separating set of $G_1\oplus G_2$ is a
minimum separating set of $G_2$, and since $\univ(G_2)$, we are done.

For the fourth item, assume that $\emp(G_1\bowtie G_2)$. Again if none
of $G_1,G_2$ is $K_1$ there must be a code vertex in each part, a
contradiction. Assume $G_1=K_1$. If $\univ(G_2)$, by
Lemma~\ref{prop:induct_sep}, $\sepID(G_1\bowtie
G_2)=\sepID(G_1)+\sepID(G_2)+1$. Since $\emp(G_1\bowtie G_2)$, the vertex
of $K_1$ cannot belong to any minimum separating set. Consider a
minimum separating set $S_2$ of $G_2$: since $\univ(G_2)$, there is a vertex $x$ of $G_2$, which is dominated by the whole set $S_2$. But since $G$ is twin-free, $x$
has a non-neighbour $y$ in $G_2$ (and $y\notin S_2$). 
Then $S_2\cup\{y\}$ is a (minimum) separating set of $G_1\bowtie G_2$. Since $\emp(G_1\bowtie G_2)$, there is a vertex $u$, necessarily in $G_2$, that is not dominated by $S_2\cup\{y\}$.
If $x$ is not adjacent to $u$, we could choose $u$ to be $y$ and $S_2\cup\{u\}$ would be a (minimum)
separating set of $G_1\bowtie G_2$ without any non-dominated
vertex (since $S_2$ is a separating set for $G_2$, there is at most one vertex of $G_2$ having no neighbours in $S_2$), a contradiction. 
Hence, $x$ is adjacent to $u$ and $u\neq y$ and $y$ is not adjacent to $u$ and $x$.
But since $u\neq y$ and $S_2$ is a separating set of $G_2$, in order to be separated from $u$, $y$ must be adjacent to some vertex $s$ of $S_2$. Then, $y,s,x,u$ induce a path on
four vertices, a contradiction since we assumed $G_1\bowtie G_2$ is a
cograph. Hence $\neg\univ(G_2)$. 
Now, if we also have $\neg\emp(G_2)$, using induction hypothesis and item 6, there is a minimum
separating set $S_2$ of $G_2$ that dominates each vertex of $G_2$ and such that no vertex of $S_2$ is dominated by all the other vertices of $S_2$. Hence $S_2$ is also a separating set of $G_1$ and $\neg\emp(G_1\bowtie G_2)$, a contradiction. 
For the converse, assume $G=K_1$, $\neg\univ(G_2)$ and $\emp(G_2)$. Then by Lemma~\ref{prop:induct_sep}
$\sepID(G_1\bowtie G_2)=\sepID(G_1)+\sepID(G_2)=\sepID(G_2)$. Let $S$ be a
minimum separating set of $G_1\bowtie G_2$: then $S\setminus V(K_1)$
is a minimum separating set of $G_2$, hence $S\setminus V(K_1)=S$ and thus
we have $\emp(G_1\bowtie G_2)$ (since $\emp(G_2)$).

For the fifth item, suppose that $\neg\univ(G_1)$ and
$\neg\univ(G_2)$. Then, by Lemma~\ref{prop:induct_sep},
$\sepID(G_1\bowtie G_2)=\sepID(G_1)+\sepID(G_2)$ and the restriction
of a separating set $S$ of $G_1\bowtie G_2$ to $G_i$ ($i\in\{1,2\}$)
is a separating set of $G_i$. Since none of $G_1,G_2$ is $K_1$, there
is vertex of $S$ in each part, hence we have $\neg\univ(G_1\bowtie
G_2)$. For the converse, assume that $\univ(G_1)$ or $\univ(G_2)$. If
both $\univ(G_1)$ and $\univ(G_2)$, by Lemma~\ref{prop:induct_sep},
$\sepID(G_1\bowtie G_2)=\sepID(G_1)+\sepID(G_2)+1$. Again, since the
restriction of a separating set $S$ of $G_1\bowtie G_2$ to $G_i$
($i\in\{1,2\}$) is a separating set of $G_i$, a minimum separating set
$S$ of $G_1\bowtie G_2$ consists of one separating set $S_1$ of $G_1$, one separating set $S_2$ of $G_2$, with an additional vertex in say $G_1$. Then the vertex in
$G_2$ that is dominated by the whole $S_2$ is also dominated by the
whole set $S$. The other case is handled similarly.

For the sixth item, we use the previous results. Assume first that $G=G_1\oplus G_2$ and that $\neg \emp(G)$ and $\neg \univ(G)$. Then we have in particular, using item 2, $\neg \emp(G_1)$ and $\neg \emp(G_2)$. Consider any minimum separating sets $S_1$ of $G_1$ and $S_2$ of $G_2$ that dominates all the vertices of $G_1$ and $G_2$ respectively. By Lemma~\ref{prop:induct_sep}, $S_1\cup S_2$ is a minimum separating set of $G$ that dominates all the vertices of $G$. Since $S_1$ and 
$S_2$ are both non-empty, $S_1\cup S_2$ has no vertex dominated by all the vertices of $S_1\cup S_2$.
Assume now that $G=G_1\bowtie G_2$ and that $\neg \emp(G)$ and $\neg \univ(G)$. Using item 5, we have $\neg \univ(G_1)$ and $\neg \univ(G_2)$. Let $S_1$ (respectively $S_2$) be a minimum separating set of $G_1$ (respectively $G_2$) with no vertex dominated by all the vertices of $S_1$ (respectively $S_2$). By Lemma~\ref{prop:induct_sep}, $S_1\cup S_2$ is a minimum separating set of $G$ and no vertex is dominated by all the vertices of $S_1\cup S_2$.
Moreover, $S_1$ and $S_2$ are both non-empty, hence $S_1\cup S_2$ dominates all the vertices of $G$. 
\end{proof}

Observe that if a cograph is twin-free, then every
intermediate cograph obtained during its construction is twin-free
too, since operations $\oplus$ and $\bowtie$ preserve twins.
This fact together with Lemmas~\ref{prop:induct_sep}
and~\ref{prop:emp-univ},  implies a linear time algorithm which constructs an identifying code of a minimum size for a given cograph (based on parsing of it cotree structure).

Moreover similar ideas
lead to an algorithm for open locating-dominating sets, the details of which are left to the reader.

\begin{theorem}
There exist linear-time algorithms that construct a minimum identifying code and a minimum open locating-dominating set of a cograph.
\end{theorem}

\subsection{Bounds for cographs}

We now use the previous discussion to give tight lower bounds on the
identifying code number, (open) locating-domination number and metric
dimension of cographs.

\begin{theorem}\label{thm:bound-cographs}
Let $G$ be a twin-free cograph on $n\geq 2$ vertices with an
identifying code of size~$k$. Then, $n\leq 2k-2$, or equivalently $\M(G)\geq
\frac{n+2}{2}$.
\end{theorem}
\begin{proof}
In fact, we prove the following stronger facts (for a cograph $G$ on at least two vertices):\\1. if $\neg\emp(G)$
and $\neg\univ(G)$, then $\sepID(G)\geq\frac{n+2}{2}$;\\
2. if $\emp(G)$ and $\neg\univ(G)$ or $\neg\emp(G)$ and $\univ(G)$,
then $\sepID(G)\geq\frac{n+1}{2}$;\\
3. if $\emp(G)$ and $\univ(G)$, then $\sepID(G)\geq\frac{n}{2}$.

The proof uses induction on the order of the cograph and the fact that
any cograph is built recursively from two cographs using operation
$\oplus$ or $\bowtie$. The claim is clearly true for the only
twin-free cograph on two vertices, $\overline{K_2}$, hence assume
$n>2$. We just have to prove the result for $G=G_1\oplus G_2$ since everything is symmetric by taking the complement and exchanging $\emp(G)$ with $\univ(G)$.

Assume first that $G_1=K_1$. Then $G_2$ has $n_2\geq 2$ vertices and by induction the properties $1,2,3$ hold for $G_2$.
We have $\emp(G_1)$ and so $\emp(G)$. If $\univ(G)$ holds, then by
Lemma~\ref{prop:emp-univ}, we have $\univ(G_2)$ and $\neg \emp(G_2)$, hence $\sepID(G)\geq \sepID(G_2)\geq \frac{n_2+1}{2} = \frac{n}{2}$, and we are done.
Assume now that $\neg \univ(G)$. If $\emp(G_2)$, then by Lemma~\ref{prop:induct_sep}, $\sepID(G)=\sepID(G_1)+\sepID(G_2)+1\geq \frac{n_2}{2}+1 \geq \frac{n+2}{2}$ and we are done.
Otherwise, we have $\neg \emp(G_2)$ and by item 3 of Lemma~\ref{prop:emp-univ}, we also have $\neg \univ(G_2)$, hence $\sepID(G)\geq\sepID(G_2)\geq \frac{n_2+2}{2}\geq \frac{n+1}{2}$.

Assume now that none of $G_1,G_2$ is $K_1$, then by induction, the properties hold for $G_1$ and $G_2$ and we have $\neg \univ(G)$. If both $\neg \emp(G_1)$ and $\neg \emp(G_2)$, then we also have $\neg \emp(G)$ and $\sepID(G)\geq \sepID(G_1)+\sepID(G_2) \geq \frac{n_1+1}{2}+\frac{n_2+1}{2}\geq \frac{n+2}{2}$ and we are done.
If both $\emp(G_1)$ and $\emp(G_2)$, then $\emp(G)$ and $\sepID(G)\geq \sepID(G_1)+\sepID(G_2)+1\geq \frac{n_1}{2}+\frac{n_2}{2}+1\geq \frac{n+1}{2}$.
Finally, if only one property holds, say $\emp(G_1)$, then $\emp(G)$ and $\sepID(G)\geq \sepID(G_1)+\sepID(G_2)\geq \frac{n_1}{2}+\frac{n_2+1}{2}\geq \frac{n+1}{2}$.
\end{proof}

\begin{proposition}
The bound of Theorem~\ref{thm:bound-cographs} is tight for infinitely
 many cographs.
\end{proposition}
\begin{proof}
We construct, inductively, graphs reaching the bound. 
Assume there are graphs 
$G_n^1$,$G_n^2$,$G_n^3$,$G_n^4$
on $n$ vertices such that
\begin{itemize}
\item $\sepID(G_n^1)= \lceil \frac{n+2}{2}\rceil $, $\neg \emp(G_n^1)$ and $\neg \univ(G_n^1)$;
\item $\sepID(G_n^2)=\lceil \frac{n+1}{2} \rceil$, $\emp(G_n^2)$ and $\neg \univ(G_n^2)$;
\item $\sepID(G_n^3)=\lceil \frac{n+1}{2} \rceil$, $\neg \emp(G_n^3)$ and $\univ(G_n^3)$;
\item $\sepID(G_n^4)=\lceil \frac{n}{2} \rceil$, $\emp(G_n^4)$, $\univ(G_n^4)$ and $G_n^4$ does not have a universal vertex.
\end{itemize}

Then the graphs $G_{n+1}^1=\overline{K_3}\bowtie G_{n-2}^2 $,
$G_{n+1}^2=K_1\oplus G_{n}^4 $, $G_{n+1} ^3= K_1\bowtie G_n^4$,
$G_{n+1}^4=K_1\oplus G_n^3$, satisfy the properties for $n+1$
vertices.

Starting with $G_3^2=\overline{K_3}$, $G_3^3=P_3$, $G_4^2=\overline{K_4}$
and $G_4^3=C_4$, we obtain a sequence of graphs $G_n^i$ for $i\geq 2$ and
$n\geq 4$ satisfying the properties. We obtain the graphs $G_n^1$ for $n\geq 6$ using the graphs $G_{n-2}^2$.
\end{proof}

We can prove similar results for locating-dominating sets. Since the proofs are 
very similar to that of identifying codes, we defer them to 
Appendix~\ref{appdx:cograph-LD}.

\begin{theorem}\label{thm:bound-cographs-MD}
Let $G$ be a connected cograph on $n\geq 2$ vertices, having a resolving set
of size~$k$ and a locating-dominating set of size $d$. Then, $n\leq
3k\leq 3d$.
\end{theorem}

 The bound of Theorem~\ref{thm:bound-cographs-MD} is tight for infinitely
 many cographs.

\begin{proposition}\label{prop:bound-cographs-MD-tight}
There exist infinitely many cographs where both inequalities of Theorem~\ref{thm:bound-cographs-MD} are simultaneously achieved.
\end{proposition}

\section{Conclusion}\label{sec:conclu}

We conclude with some open problems. It would be interesting to know whether similar bounds, as the ones we established here, hold for other standard graph classes. One specific case that is worth studying is the metric dimension of planar graphs and line graphs. For these two classes, such bounds are known to exist for locating-dominating sets and identifying codes ($O(k)$ for planar graphs~\cite{RS84} and $O(k^2)$ for line graphs~\cite{lineID}). Do bounds of the form $O(Dk)$ and $O(Dk^2)$, respectively, hold for planar graphs and line graphs? \footnote{After the completion and submission of this paper, this has been investigated in~\cite{BDFHMP16} by two of the authors with other colleagues. The answer turns out to be negative for both planar graphs and line graphs. In that paper, it is shown that there exist line graphs of diameter~$4$ and order $\Omega(2^k)$, outerplanar graphs of order $\Theta(kD^2)$, and planar graphs of metric dimension~$3$ and order~$\Theta(D^3)$. Nevertheless, for planar graphs, it is proved there that $n=O(D^4k^4)$ holds.}

We also remark that during the writing of this paper, the fourth author, together with several colleagues~\cite{BLLPT14}, proved that for any graph $G$ of order $n$ and VC-dimension at most~$d$, the bound $n\leq O(k^d)$ holds, where $k$ is the size of an identifying code of $G$ (the same bound also applies to (open) locating-dominating sets). In particular, interval graphs have VC-dimension at most~2, and permutation graphs have VC-dimension at most~3. Hence, the result of~\cite{BLLPT14} generalizes some of our results (but our bounds are more precise).

\appendix

\section{Locating-dominating sets and metric dimension of cographs}\label{appdx:cograph-LD}

As mentioned in the introduction, if $G$ has diameter~2, we have
$\MD(G)\leq\LD(G)\leq\MD(G)+1$, where the upper bound is reached if
and only if for every smallest resolving set there is a non-dominated
vertex. Since we will use the cotree structure of cographs, we have to deal with not connected graphs for which the difference between $\MD(G)$ and $\LD(G)$ can be more than one. As before, we denote by $\sepLD(G)$ the smallest size of a {\em separating set}, that is, a set $S\subseteq V(G)$ that separates all the vertices of $V(G)\setminus S$ (it is a locating-dominating set without the condition of being a dominating set). If $G$ is a connected cograph, it has diameter $2$ and a separating set is equivalent to a resolving set, in particular $\sepLD(G)=\MD(G)$. If $G$ is not connected, then the two parameters can be different since in a resolving set, one vertex per connected component could be not dominated. We define $\empMD(G)$ as the property that for a graph $G$, every minimum separating set leaves a non-dominated vertex; $\univMD(G)$ is
the property that every minimum separating set $S$ of $G$ leaves a
vertex in $G\setminus S$ that is dominated by all vertices of $S$. We have $\sepLD(G)\leq \LD(G)\leq \sepLD(G)+1$ and $\LD(G)=\sepLD(G)+1$ if and only if $\empMD(G)$.

Note that $S$ is a separating set of $G$ if and only if it is a separating set of the complement of $G$. Moreover, the following hold:\\
\indent $\empMD(G)$ if and only if $\univMD(\overline{G})$\\
\indent $\univMD(G)$ if and only if $\empMD(\overline{G})$

\medskip

We have the following lemma.

\begin{lemma}\label{prop:induct_sep_MD}
Let $G_1,G_2$ be two cographs with $\sepLD(G_1)=k_1$ and
$\sepLD(G_2)=k_2$. Then, $k_1+k_2\leq\sepLD(G_1\oplus G_2)\leq k_1+k_2+1$,
where the upper bound is reached if and only if we have $\empMD(G_1)$
and $\empMD(G_2)$. Moreover, $k_1+k_2\leq\sepLD(G_1\bowtie G_2)\leq
k_1+k_2+1$ and the upper bound is reached if and only if we have
$\univMD(G_1)$ and $\univMD(G_2)$.
\end{lemma}
\begin{proof}
Note that in both $G_1\oplus G_2$ and $G_1\bowtie G_2$, a vertex in
$G_1$ cannot separate a pair in $G_2$, and vice-versa. Hence, for
every separating set of $G_1\oplus G_2$ or $G_1\bowtie G_2$, its
restriction to $G_i$ for $i\in\{1,2\}$ is a separating set of
$G_i$. This proves the two lower bounds.

For the upper bounds, let $S_1$ and $S_2$ be minimum separating sets
of $G_1$ and $G_2$, respectively. If $S=S_1\cup S_2$ is not a
separating set of $G_1\oplus G_2$ or $G_1\bowtie G_2$, by the previous
observation, there must be a pair $u,v$ with $u\in G_1$ and $v\in G_2$
that is not separated. In the case of $G_1\oplus G_2$, these two
vertices must both be non-dominated by $S$, and this is the only
non-separated pair. Then, adding one of them gives a separating set of
size $k_1+k_2+1$. For the case $G_1\bowtie G_2$, $u$ is dominated
by all vertices of $S_2$, and $v$ is dominated by all vertices of
$S_1$. Hence both $u,v$ must be dominated by all vertices of $S$, and
this is the only non-separated pair. Then $S\cup\{u\}$ is a resolving
set of $G_1\bowtie G_2$ of size $k_1+k_2+1$.
\qed\end{proof}

Using the following lemma, it is easy to keep track of the properties
$\empMD$ and $\univMD$ while parsing the cotree structure of a cograph
$G$.

\begin{lemma}\label{prop:emp-univ-MD} We have:\\
1. $\empMD(K_1)$ and $\univMD(K_1)$;\\
2. $\empMD(G_1\oplus G_2)$ if and only if $\empMD(G_1)$ or $\empMD(G_2)$;\\
3. $\univMD(G_1\oplus G_2)$ if and only if one of $G_1,G_2$ (say
$G_1$) is $K_1$, $\univMD(G_2)$ and $\neg\empMD(G_2)$;\\
4. $\empMD(G_1\bowtie G_2)$ if and only if $G_1=K_1$, $\neg\univMD(G_2)$ and $\empMD(G_2)$;\\
5. $\univMD(G_1\bowtie G_2)$ if and only if $\univMD(G_1)$ or $\univMD(G_2)$.\\
6. If $\neg \empMD(G)$ and $\neg \univMD(G)$, there is a minimum separating set $S$ of $G$ such that all the vertices are dominated by some vertex of $S$ and there is no vertex out of $S$ that dominates all $S$.
\end{lemma}
\begin{proof}
Since taking the complement is the same as exchanging $\empMD$ to $\univMD$ and $\bowtie$ to $\oplus$, we just have to prove items 1, 2, 3 and 6. We prove these items by induction.

The first item is clear. For the second item, assume $\empMD(G_1\oplus
G_2)$. By Lemma~\ref{prop:induct_sep_MD}, if $\neg\empMD(G_1)$ or
$\neg\empMD(G_2)$, then any minimum separating set of $G_1\oplus G_2$ is
the union of a minimum separating set of $G_1$ and one of $G_2$. Hence
if both $\neg\empMD(G_1)$ and $\neg\empMD(G_2)$, then $\neg\empMD(G_1\oplus
G_2)$. Now, if $\empMD(G_1)$ and $\neg\empMD(G_2)$ (or vice-versa), for
the same reason we have $\empMD(G_1\oplus G_2)$. If both $\empMD(G_1)$ and
$\empMD(G_2)$, then again by Lemma~\ref{prop:induct_sep_MD},
$\sepLD(G_1\oplus G_2)=\sepLD(G_1)+\sepLD(G_2)+1$, but since no vertex of
$G_1$ dominates any vertex of $G_2$ (and vice-versa), there must
remain a non-dominated vertex.

For the third item, assume $\univMD(G_1\oplus G_2)$. If none of
$G_1,G_2$ is $K_1$, then there must be a code vertex in both
$G_1,G_2$, hence $\neg\univMD(G_1\oplus G_2)$. Hence assume $G_1=K_1$,
and let $S_2$ be a minimum separating set of $G_2$. By
Lemma~\ref{prop:induct_sep_MD}, if $\empMD(G_2)$, $\sepLD(G_1\oplus
G_2)=\sepLD(G_1)+\sepLD(G_2)+1$. But then $S'=S_2\cup V(K_1)$ is a minimum
separating set of $G_1\oplus G_2$ without a vertex dominated by the
whole of $S'$, a contradiction. Hence, $\neg\empMD(G_2)$. If
$\neg\univMD(G_2)$, using induction and item 6, there is a minimum separating set $S_2$ of $G_2$
without any vertex dominated by the whole of $S_2$ and all the vertices dominated by some vertex of $S_2$.
Then $S_2$ is a minimum
separating set of $G$ without any vertex dominated by the whole set $S_2$, a contradiction.
Hence $\univMD(G_2)$. For the converse, if $G_1=K_1$, $\univMD(G_2)$ and
$\neg\empMD(G_2)$, then by Lemma~\ref{prop:induct_sep_MD}, $\sepLD(G_1\oplus
G_2)=\sepLD(G_1)+\sepLD(G_2)$, and it is clear that the vertex of $K_1$
does not belong to a any minimum separating set of $G_1\oplus
G_2$. Hence every minimum separating set of $G_1\oplus G_2$ is a
minimum separating set of $G_2$, and we are done.

For the sixth item, we use the previous results. Assume first that $G=G_1\oplus G_2$ and that $\neg \empMD(G)$ and $\neg \univMD(G)$. Then we have in particular, using item 2, $\neg \empMD(G_1)$ and $\neg \empMD(G_2)$. Consider any minimum separating sets $S_1$ of $G_1$ and $S_2$ of $G_2$ that dominates all the vertices of $G_1$ and $G_2$ respectively. By Lemma~\ref{prop:induct_sep_MD}, $S_1\cup S_2$ is a minimum separating set of $G$ that dominates all the vertices of $G$. Since $S_1$ and 
$S_2$ are both non-empty, $S_1\cup S_2$ has no vertex dominated by all the vertices of $S_1\cup S_2$.
Assume now that $G=G_1\bowtie G_2$ and that $\neg \empMD(G)$ and $\neg \univMD(G)$. Using item 5, we have $\neg \univMD(G_1)$ and $\neg \univMD(G_2)$. Let $S_1$ (respectively $S_2$) be a minimum separating set of $G_1$ (respectively $G_2$) with no vertex dominated by all the vertices of $S_1$ (respectively $S_2$). By Lemma~\ref{prop:induct_sep_MD}, $S_1\cup S_2$ is a minimum separating set of $G$ and no vertex is dominated by all the vertices of $S_1\cup S_2$.
Moreover, $S_1$ and $S_2$ are both non-empty, hence $S_1\cup S_2$ dominates all the vertices of $G$.
\qed\end{proof}

We can now prove Theorem~\ref{thm:bound-cographs-MD}.

\begin{proof}[Proof of Theorem~\ref{thm:bound-cographs-MD}]
In fact, we prove the following stronger facts:
\\1. if $\neg\empMD(G)$
and $\neg\univMD(G)$, $\sepLD(G)\geq\frac{n}{3}$;\\
2. if $\empMD(G)$ and
$\neg\univMD(G)$ or $\neg\empMD(G)$ and $\univMD(G)$,
$\sepLD(G)\geq\frac{n+1}{3}$;\\3. if $\empMD(G)$ and $\univMD(G)$,
$\sepLD(G)\geq\frac{n+2}{3}$.

The claim is clearly true for $K_2$ and $\overline{K_2}$, hence assume
$n>2$. We just have to prove the result for $G=G_1\oplus G_2$ since
everything is symmetric by taking the complement and exchanging
$\empMD(G)$ with $\univMD(G)$.

Assume first that $G_1=K_1$. Then $G_2$ has at least $n_2\geq 2$
vertices and by induction the properties $1,2,3$ hold for $G_2$. We
have $\empMD(G_1)$ and so $\empMD(G)$. If $\univMD(G)$ holds, then by
Lemma~\ref{prop:emp-univ-MD}, we have $\univMD(G_2)$ and $\neg
\empMD(G_2)$, hence $\sepLD(G)\geq \sepLD(G_2)\geq \frac{n_2+1}{3}\geq
\frac{n}{3}$, and we are done. Assume now that $\neg \univMD(G)$. If
$\empMD(G_2)$, then by Lemma~\ref{prop:induct_sep_MD},
$\sepLD(G)=\sepLD(G_1)+\sepLD(G_2)+1\geq \frac{n_2}{3}+1 \geq \frac{n+2}{3}$
and we are done. Otherwise, we have $\neg \empMD(G_2)$ and by
Lemma~\ref{prop:emp-univ-MD}, we also have $\neg \univMD(G_2)$, hence
$\sepLD(G)\geq\sepLD(G_2)\geq \frac{n_2+2}{3}= \frac{n+1}{3}$.

Assume now that none of $G_1,G_2$ is $K_1$, then by induction, the
properties hold for $G_1$ and $G_2$ and we have $\neg \univMD(G)$. If
both $\neg \empMD(G_1)$ and $\neg \empMD(G_2)$, then we also have
$\neg \empMD(G)$ and $\sepLD(G)\geq \sepLD(G_1)+\sepLD(G_2) \geq
\frac{n_1+1}{3}+\frac{n_2+1}{3}\geq \frac{n+2}{3}$ and we are done.
If both $\empMD(G_1)$ and $\empMD(G_2)$, then $\empMD(G)$ and $\sepLD(G)=
\sepLD(G_1)+\sepLD(G_2)+1\geq \frac{n_1}{3}+\frac{n_2}{3}+1= \frac{n+3}{3}$.
Finally, if only one, say $\empMD(G_1)$ holds, then $\empMD(G)$ and
$\sepLD(G)\geq \sepLD(G_1)+\sepLD(G_2)\geq \frac{n_1}{3}+\frac{n_2+1}{3}\geq
\frac{n+1}{3}$.
\end{proof}

We now prove Proposition~\ref{prop:bound-cographs-MD-tight}:

\begin{proof}[Proof of Proposition~\ref{prop:bound-cographs-MD-tight}]
We construct graphs reaching the bound by induction
as follows. Assume there exist graphs $G_n^1$, $G_n^2$, $G_n^3$, $G_n^4$
on $n$ vertices such that
\begin{itemize}
\item $\sepLD(G_n^1)= \lceil \frac{n+2}{3}\rceil $, $\neg \empMD(G_n^1)$ and $\neg \univMD(G_n^1)$;
\item $\sepLD(G_n^2)=\lceil \frac{n+1}{3} \rceil$, $\empMD(G_n^2)$ and $\neg \univMD(G_n^2)$;
\item $\sepLD(G_n^3)=\lceil \frac{n+1}{3} \rceil$, $\neg \empMD(G_n^3)$ and $\univMD(G_n^3)$;
\item $\sepLD(G_n^4)=\lceil \frac{n}{3} \rceil$, $\empMD(G_n^4)$, $\univMD(G_n^4)$ and $G_n^4$ does not have a universal vertex.
\end{itemize}

Then the graphs $G_{n+1}^1=K_2\oplus G_{n-1}^3 $,
$G_{n+1}^2=K_1\oplus G_{n}^1 $, $G_{n+1}^3= K_1\bowtie G_n^1$,
$G_{n+1}^4=K_1\oplus G_n^3$, satisfy the properties for $n+1$
vertices.

Starting with $G_2^2=\overline{K_2}$, $G_2^3=K_2$,
$G_3^2=\overline{K_3}$, $G_3^3=K_3$, $G_4^2=\overline{K_2}\oplus K_2$
and $G_4^3=K_1\bowtie (K_1\oplus K_2)$, we obtain $G_4^1$, $G_3^4$ and
$G_4^4$ and then graphs $G_n^i$ for $n\geq 5$ satisfying the
properties.
\end{proof}

\end{document}